\documentclass[%
reprint,
superscriptaddress,
nofootinbib,
amsmath,amssymb,
aps,longbibliography
]{revtex4-1}

\newif\ifarxiv 
\arxivtrue

\usepackage{amsthm,bm,xparse,xcolor,tikz}
\usepgfmodule{decorations}
\usetikzlibrary{decorations,decorations.pathmorphing,patterns}
\usepackage[T1]{fontenc}

\AtBeginDocument{%
    \newwrite\bibnotes
    \def\bibnotesext{Notes.bib}
    \immediate\openout\bibnotes=\jobname\bibnotesext
    \immediate\write\bibnotes{@CONTROL{REVTEX41Control}}
    \immediate\write\bibnotes{@CONTROL{%
    apsrev41Control,author="08",editor="1",pages="1",title="0",year="0"}}
     \if@filesw
     \immediate\write\@auxout{\string\citation{apsrev41Control}}%
    \fi
  }%

\usepackage[nodayofweek]{datetime}  
\usepackage{hyperref}
\usepackage{xcolor}
\usepackage[capitalise]{cleveref}
\usepackage{enumerate}   
\definecolor{mylinkcolor}{rgb}{0,0,0.8} 
\hypersetup{unicode=true, %
  bookmarksnumbered=false,bookmarksopen=false,bookmarksopenlevel=1, %
  breaklinks=true,pdfborder={0 0 0},colorlinks=true}%
\hypersetup{%
  anchorcolor=mylinkcolor,citecolor=mylinkcolor, %
  filecolor=mylinkcolor,linkcolor=mylinkcolor, %
  menucolor=mylinkcolor,runcolor=mylinkcolor, %
  urlcolor=mylinkcolor}%

\newtheorem{theorem}{Theorem}
\newtheorem{lemma}{Lemma}
\newtheorem{corollary}{Corollary}

\theoremstyle{definition}
\newtheorem{proposition}{Proposition} 
\newtheorem{definition}{Definition}

\newcommand{\ket}[1]{| #1 \rangle}
\newcommand{\bra}[1]{\langle #1 |}

\newcommand{\ketbra}[2]{|#1\rangle\!\langle#2|}

\newcommand{\cH}{\mathcal{H}}

\begin{document}
\title{Tight analytic bound on the trade-off between device-independent randomness and nonlocality}
\author{Lewis Wooltorton}
    \email{lewis.wooltorton@york.ac.uk}
    \affiliation{Department of Mathematics, University of York, Heslington, York, YO10 5DD, United Kingdom}
    \affiliation{Quantum Engineering Centre for Doctoral Training, H. H. Wills Physics Laboratory and Department of Electrical \& Electronic Engineering, University of Bristol, Bristol BS8 1FD, United Kingdom}
\author{Peter Brown}
    \email{peter.brown@telecom-paris.fr}
    \affiliation{Télécom Paris, LTCI, Institut Polytechnique de Paris,
19 Place Marguerite Perey, 91120 Palaiseau, France}
\author{Roger Colbeck}
    \email{roger.colbeck@york.ac.uk}
    \affiliation{Department of Mathematics, University of York, Heslington, York, YO10 5DD, United Kingdom}

\date{$5^{\text{th}}$ October 2022}

\begin{abstract}
Two parties sharing entangled quantum systems can generate correlations that cannot be produced using only shared classical resources. These \emph{nonlocal} correlations are a fundamental feature of quantum theory but also have practical applications. For instance, they can be used for \emph{device-independent} (DI) random number generation, whose security is certified independently of the operations performed inside the devices. The amount of certifiable randomness that can be generated from some given non-local correlations is a key quantity of interest. Here we derive tight analytic bounds on the maximum certifiable randomness as a function of the nonlocality as expressed using the Clauser-Horne-Shimony-Holt (CHSH) value. We show that for every CHSH value greater than the local value ($2$) and up to $3\sqrt{3}/2\approx2.598$ there exist quantum correlations with that CHSH value that certify a maximal two bits of global randomness. Beyond this CHSH value the maximum certifiable randomness drops.  We give a second family of Bell inequalities for CHSH values above $3\sqrt{3}/2$, and show that they certify the maximum possible randomness for the given CHSH value. Our work hence provides an achievable upper bound on the amount of randomness that can be certified for any CHSH value. We illustrate the robustness of our results, and how they could be used to improve randomness generation rates in practice, using a Werner state noise model.
\end{abstract}
\maketitle

\ifarxiv\section{Introduction}\label{sec:intro}\else\noindent{\it Introduction.|}\fi
Nonlocality is the phenomenon where measurements of certain quantum systems, by isolated observers, generate correlations inaccessible to any local systems that behave classically~\cite{Bell_book,Brunner_review}. Nonlocal correlations can be used to make statements about the underlying quantum system without characterizing the devices used~\cite{mayers2004self,McKagueSinglet,YangSelfTest,KaniewskiSelfTest,SupicSelfTest}, and constitute a resource for information processing~\cite{BarrettNonlocalResource}. In particular, they give rise to the possibility of DI information processing which allows, for instance, the intrinsic randomness of nonlocal correlations to be exploited for randomness expansion~\cite{ColbeckThesis,PAMBMMOHLMM,CK2,MS1,MS2}, amplification~\cite{CR_free}, and key distribution protocols~\cite{Ekert,BHK,ABGMPS,PABGMS,VV2,ADFRV}.

Given some experimental conditions in a particular input-output scenario, what is the optimal way to generate randomness device-independently? Since the values of extremal Bell inequalities quantify the distance of the observed correlations from the local boundary, one might expect these to be optimal for randomness. However, the relationship between nonlocality and maximum randomness is nontrivial~\cite{delaTorreMaxNonlocal}, and it has been shown that non-extremal Bell inequalities can certify more randomness in some cases~\cite{AcinRandomnessNonlocality}.

A substantial literature has developed investigating the maximum achievable randomness in different DI scenarios. In particular, the existence of Bell tests that can certify maximum global randomness was shown in~\cite{DharaMaxRand} by adding extra measurements. Constructions achieving maximal randomness in the bipartite scenario for non-projective measurements were given in~\cite{AcinOptimalEbit} and for greater than two projective measurements per party in~\cite{Law_2014,Andersson18,BRC,WoodheadMaxRandomness}. In~\cite{AcinRandomnessNonlocality} a construction that tends towards the maximum 2 random bits is presented, based on the violation of tilted-CHSH inequalities~\cite{AcinRandomnessNonlocality,BampsPironio}. This provides a key example where non-extremal Bell inequalities certify more randomness (maximum violation of the CHSH inequality, the only extremal Bell inequality in the 2-input, 2-output scenario up to symmetry, can certify $5/2-\log_2(1+\sqrt{2})/\sqrt{2}\approx 1.601$ bits of global randomness by comparison --- see, e.g.,~\cite{BhavsarDI}). Based on~\cite{AcinRandomnessNonlocality} one might expect that achieving 2 bits of randomness requires the CHSH violation (or entanglement) of the strategy to tend to $0$. If this were the case, there would be a problem with the robustness of the construction, and the result suggests a trade-off between certifiable randomness and distance from the local set. Ref.~\cite{AcinRandomnessNonlocality} left open whether two bits of randomness is actually attainable using a single statistic in the 2-input 2-output scenario, and how non-local a strategy achieving this can be.

Our work gives conclusive answers to these questions. We consider the maximum amount of DI randomness that can be certified from the set of quantum correlations achieving a particular CHSH value. In other words, we investigate how much randomness is achievable when the generating system is required to exhibit a particular amount of nonlocality. To do so, we introduce two families of Bell expressions that self-test families of two qubit strategies. Our first family (see Proposition \ifarxiv\ref{prep:delta}\else 1\fi) certifies exactly 2 random bits for all CHSH values in the interval $(2,3\sqrt{3}/2]$, showing 2 random bits are achievable without tending towards the local set~\cite{AcinRandomnessNonlocality}, requiring extra measurements~\cite{BRC,WoodheadMaxRandomness,Law_2014} or constraining the full distribution~\cite{BrownDeviceIndependent2}. Our second family (see Proposition \ifarxiv\ref{prep:gamma}\else 2\fi) covers the range of values $[3\sqrt{3}/2,2\sqrt{2}]$, coinciding with the CHSH inequality for $2\sqrt{2}$, and the certifiable randomness achieved is a smooth, monotonically decreasing function of the value. We show in Proposition \ifarxiv\ref{prep:maxRand} \else 3 \fi that this is the true maximum randomness achievable for this range of CHSH values, illustrating how one only needs to sacrifice randomness when approaching the maximum quantum value of CHSH. Finally, we analyse the robustness of our construction under a Werner state noise model~\cite{Werner}, and compare it to that of the tilted CHSH inequalities. We find both constructions to be robust, and at any given noise level there exists an optimal statistic for practical DI randomness generation that can outperform CHSH.          

\ifarxiv\section{DI scenario}\else{\medskip\noindent\it DI scenario.|}\fi We consider the bipartite 2-input 2-output Bell scenario. Let two isolated devices each receive an input $x,y\in \{0,1\}$, from which they produce an output $a,b \in \{0,1\}$, stored in the classical registers $A$ and $B$. The devices are characterised by the joint conditional probability distribution $\text{p}(ab|xy)$, which, due to the isolation of the devices, must be no-signalling.

We refer to a quantum strategy when the devices share a bipartite density operator $\rho_{Q_{A}Q_{B}}$ on the Hilbert space $\mathcal{H}_{Q_{A}} \otimes \mathcal{H}_{Q_{B}}$, and measure observables $A_{x} = M_{0|x}-M_{1|x},B_{y}=N_{0|y} - N_{1|y}$, where $\{M_{a|x}\}_{a},\{N_{b|y}\}_{b}$ are projective measurements on the associated Hilbert space (projective measurements can be assumed without loss of generality according to Naimark's dilation theorem~\cite{paulsen_2003}). We also include the possibility of an adversary Eve, who wishes to guess the outputs. In the DI scenario, Eve may have supplied the devices used by the user (Alice) and may hold a purifying system $E$ with associated Hilbert space $\mathcal{H}_{E}$ such that the post-measurement system $AB$ and $E$ are correlated. We describe this using a tripartite density operator, $\rho_{Q_{A}Q_{B}E}$ such that $\rho_{Q_{A}Q_{B}} = \text{Tr}_{E}[\rho_{Q_{A}Q_{B}E}]$. Following measurement with inputs $X=x$ and $Y=y$, we obtain the classical-quantum state $\rho_{ABE} = \sum_{ab}\ketbra{ab}{ab}_{AB} \otimes \rho_{E}^{(a,b,x,y)}$, where $\rho_{E}^{(a,b,x,y)} = \text{Tr}_{Q_{A}Q_{B}}[\rho_{Q_{A}Q_{B}E}(M_{a|x}\otimes N_{b|y} \otimes \mathbb{I}_{E})]$ is proportional to Eve's state conditioned on the joint measurement outcomes, and the distribution is recovered via $\text{p}(ab|xy) = \text{Tr}[\rho_{E}^{(a,b,x,y)}]$.

\ifarxiv\section{Nonlocality and self-tests}\else{\medskip\noindent\it Nonlocality and self-tests.|}\fi To quantify the distance of an observed distribution $\mathrm{p}(ab|xy)$ from the local boundary in this scenario, we consider the CHSH expression $I_{\text{CHSH}} = \langle A_{0}B_{0} \rangle +  \langle A_{0}B_{1} \rangle + \langle A_{1}B_{0} \rangle - \langle A_{1}B_{1}\rangle$, where $\langle A_{x}B_{y}\rangle = \sum_{ab}(-1)^{a+b} \ \text{p}(ab|xy)=\bra{\Psi}A_{x} \otimes B_{y} \otimes \mathbb{I} \ket{\Psi}$ when $\mathrm{p}(ab|xy)$ admits a quantum representation with purified state and observables $\big(\ket{\Psi}_{Q_{A}Q_{B}E},A_{x},B_{y}\big)$. The local bound is given by $I_{\text{CHSH}} \leq 2$, and the maximum quantum value is $2\sqrt{2}$~\cite{Cirelson}. Any distribution that violates the local bound is said to be nonlocal. 

It is known that the CHSH inequality self-tests the maximally entangled state $(\ket{00} + \ket{11})/\sqrt{2}$ and the measurements that achieve its maximum quantum violation~\cite{PopescuRohrlich,BardynSelfTest,SupicSelfTest,BampsPironio} in the sense that there is only one state and set of measurements that can achieve $I_{\text{CHSH}}=2\sqrt{2}$ up to local isometries. One can also define a robust self test, in which close to maximum violation certifies a state and measurements close to the optimal ones up to local isometries.

\ifarxiv\section{Entropy bounds}\else{\medskip\noindent\it Entropy bounds.|}\fi The quantity of interest for calculating the DI global randomness is the conditional von Neumann entropy when the devices receive inputs $X=Y=0$, $H(AB|X=0,Y=0,E)$, evaluated for the post-measurement state $\rho_{ABE}$. This is the relevant quantity for spot-checking DI random number generation~\cite{BhavsarDI}. For DI randomness expansion we require lower bounds on this quantity that hold for all states and measurements compatible with the observed distribution $P_{\text{obs}}$, or some linear functions $f_i$ of $P_{\text{obs}}$, e.g., the CHSH value. This gives the asymptotic rate of randomness generation $r$, in bits per round:
\begin{equation}
    r = \!\!\!\!\!\!\inf_{\substack{ \rho_{Q_{A}Q_{B}E}, \\ \{M_{a|x}\}_{a}, \{N_{b|y}\}_{b} \\ \text{compatible with} \ f_i(P_{\text{obs}})}}\!\!\!\!\!\! H(AB|X=0,Y=0,E)_{\rho_{ABE}}. \label{eq:rate}
\end{equation}
The asymptotic rate can also be used as a basis for rates with finite statistics using tools such as the entropy accumulation theorem~\cite{DFR,EAT2,LLR&}.  

In the noiseless scenario, we prove a self-testing statement that certifies a state and measurements that generate two bits of randomness. In this case, $f(P_{\text{obs}})$ is a self-testing Bell expression, and there is only one state and set of measurements that can achieve the maximal quantum value (up to symmetries), from which the conditional entropy can be evaluated.  For the noisy case, we use the recently developed numerical technique from~\cite{BrownDeviceIndependent2} to compute lower bounds on \cref{eq:rate} using semidefinite programming. 

\ifarxiv\section{Main results}
\else{\medskip\noindent\it Main results.|}\fi Our first main result is the family of Bell expressions that works for CHSH values in the range $(2,3\sqrt{3}/2]$.

\ifarxiv
\begin{proposition}
Let $0 < \delta \leq \pi/6$, and define the family of Bell expressions parameterized by $\delta$, labelled $I_{\delta}$:
\begin{multline}
    \langle A_{0} B_{0} \rangle\!\! +\!\! \frac{1}{\sin \delta}\left(  \langle A_{0} B_{1} \rangle\! +\! \langle A_{1} B_{0} \rangle \right)\!\!   
    -\!\!\frac{1}{\cos 2\delta} \langle A_{1} B_{1} \rangle . \label{eq:deltaIneq_1}
\end{multline}
Then we have the following:
\begin{enumerate}[(i)]
    \item The local bound is given by $I_{\delta}^{\mathrm{L}} = -1 + \frac{2}{\sin{\delta}} + \frac{1}{\cos{2\delta}}$.
    \item The quantum bound is given by $I_{\delta}^{\mathrm{Q}} = \frac{2\cos^{3}\delta}{\cos 2\delta \sin \delta}$.
    \item Up to local isometries there is a unique strategy that achieves $I_{\delta} = I_{\delta}^{\mathrm{Q}}$:
\begin{align}
    \rho_{Q_AQ_B}&=\ket{\psi}\bra{\psi}\text{ where }\ket{\psi} = \frac{1}{\sqrt{2}}(\ket{00} + \ket{11}), \nonumber \\
    A_{0} &= \sigma_{Z}, \ \ \ B_{0} = \sigma_{X}, \nonumber \\
    A_{1} &= -\sin \delta \, \sigma_{Z} + \cos \delta \, \sigma_{X}, \nonumber \\ B_{1} &= \cos \delta \, \sigma_{Z} - \sin \delta \, \sigma_{X}.
    \label{eq:deltaStrat_1}
\end{align}
\end{enumerate}

\label{prep:delta}
\end{proposition} 
\else{\medskip \noindent \textbf{Proposition 1.} \textit{Let $0 < \delta \leq \pi/6$, and define the family of Bell expressions parameterized by $\delta$, labelled $I_{\delta}$:
\begin{equation}
    \langle A_{0} B_{0} \rangle + \frac{1}{\sin \delta}\left(  \langle A_{0} B_{1} \rangle\! +\! \langle A_{1} B_{0} \rangle \right)   
    -\frac{1}{\cos 2\delta} \langle A_{1} B_{1} \rangle . \label{eq:deltaIneq_1}
\end{equation}
Then we have the following:\\
\noindent(i) The local bound is given by $I_{\delta}^{\mathrm{L}} = -1 + \frac{2}{\sin{\delta}} + \frac{1}{\cos{2\delta}}$.\\
\noindent(ii) The quantum bound is given by $I_{\delta}^{\mathrm{Q}} = \frac{2\cos^{3}\delta}{\cos 2\delta \sin \delta}$.\\
\noindent(iii) Up to local isometries there is a unique strategy that achieves $I_{\delta} = I_{\delta}^{\mathrm{Q}}$:
\begin{align}
    \rho_{Q_AQ_B}&=\ket{\psi}\bra{\psi}\text{ where }\ket{\psi} = \frac{1}{\sqrt{2}}(\ket{00} + \ket{11}), \nonumber \\
    A_{0} &= \sigma_{Z}, \ \ \ B_{0} = \sigma_{X}, \nonumber \\
    A_{1} &= -\sin \delta \, \sigma_{Z} + \cos \delta \, \sigma_{X}, \nonumber \\ B_{1} &= \cos \delta \, \sigma_{Z} - \sin \delta \, \sigma_{X}.
    \label{eq:deltaStrat_1}
\end{align}
\label{prep:delta}}
}
\fi

By the previous discussion of the noiseless case, Proposition \ifarxiv\ref{prep:delta} \else 1 \fi shows that there exists a family of two-qubit strategies that can achieve exactly two bits of global DI randomness in the bipartite, 2-input 2-output case; this follows from self-testing measurements $A_{0},B_{0}$ together with the maximally entangled state $\ket{\psi}$. We now explore some implications of this. The strategy in \cref{eq:deltaStrat_1} has a CHSH value $I_{\text{CHSH}} = 2\cos\delta\,(\sin \delta + 1)$ and by sweeping $0<\delta\leq \pi/6$ the interval of values $(2,3\sqrt{3}/2]$ is achieved. Hence for every CHSH value in this interval, there exists a two-qubit strategy achieving this value, that can certify exactly 2 bits of randomness. In fact $3\sqrt{3}/2 \approx 2.598$ is the largest CHSH value for which exactly two bits of randomness can be achieved, corresponding to the $\delta = \pi/6$ strategy in \cref{eq:deltaStrat_1} (see the Supplemental Material~\cite{supp}). This strategy can be derived by fixing $\text{p}(ab|00) = 1/4$, and optimising the the remaining measurement angles for the maximum CHSH value. This improves upon the results in~\cite{AcinRandomnessNonlocality} (cf.\ the introduction): rather than needing low CHSH violation to get close to maximal randomness, maximum randomness is achieved well into the nonlocal region.  

Next we derive the maximum randomness for strategies achieving a CHSH value in the interval $[3\sqrt{3}/2,2\sqrt{2}]$. Up to local isometries, the only strategy that can achieve $I_{\text{CHSH}} = 2\sqrt{2}$ is given by the maximally entangled state with measurements $A_{0} = \sigma_{Z}$, $B_{0} = (\sigma_{Z}+\sigma_{X})/\sqrt{2}$, $A_{1} = \sigma_{X},$ and $B_{1} = (\sigma_{Z}-\sigma_{X})/\sqrt{2}$, since the CHSH inequality self-tests this state and measurements~\cite{BampsPironio}. This strategy gives roughly $1.601$ bits of randomness. There must therefore be a transition between $I_{\text{CHSH}} = 3\sqrt{3}/2$ and $I_{\text{CHSH}} = 2\sqrt{2}$, where in order to achieve a larger CHSH value, randomness must be sacrificed. This transition is given by the following proposition.

\ifarxiv
\begin{proposition}
Let $0 \leq \gamma \leq \pi/12$, and define the family of Bell expressions parameterized by $\gamma$, labelled $J_{\gamma}$:
\begin{multline}
    \langle A_{0} B_{0} \rangle + c(\gamma)\big(  \langle A_{0} B_{1} \rangle  + \langle A_{1} B_{0} \rangle - \langle A_{1} B_{1} \rangle \big), \label{eq:gammaIneq_1}
\end{multline}
where $c(\gamma) = 4\cos^{2} \left(\gamma + \frac{\pi}{6}\right) - 1$. Then we have the following:
\begin{enumerate}[(i)]
    \item The local bound is given by $J_{\gamma}^{\mathrm{L}} = 12\cos^2\left( \gamma + \frac{\pi}{6} \right) - 4$.
    \item The quantum bound is given by $J_{\gamma}^{\mathrm{Q}} = 8\cos^3\left( \gamma + \frac{\pi}{6}\right)$.
    \item Up to local isometries there is a unique strategy that achieves $J_{\gamma} = J_{\gamma}^{\mathrm{Q}}$:
\begin{gather}
    \rho_{Q_AQ_B}=\ket{\psi}\bra{\psi}\text{ where }\ket{\psi} = \frac{1}{\sqrt{2}}(\ket{00} + \ket{11}), \nonumber \\
    A_{0} = \sigma_{Z}, \ \ \ B_{0} = \sin 3\gamma \, \sigma_{Z} +\cos 3\gamma \, \sigma_{X}, \nonumber \\
    A_{1} = \cos\left(\frac{2\pi}{3}-2\gamma \right)  \sigma_{Z} + \sin\left(\frac{2\pi}{3}-2\gamma \right)  \sigma_{X}, \nonumber \\ B_{1} = \cos\left(\frac{\pi}{6}+\gamma \right)  \sigma_{Z} - \sin \left(\frac{\pi}{6}+\gamma \right)  \sigma_{X}.
    \label{eq:gammaStrat}
\end{gather}
\end{enumerate}

\label{prep:gamma}
\end{proposition} 
\else{\medskip \noindent \textbf{Proposition 2.} \textit{Let $0 \leq \gamma \leq \pi/12$, and define the family of Bell expressions parameterized by $\gamma$, labelled $J_{\gamma}$:
\begin{equation}
    \langle A_{0} B_{0} \rangle + c(\gamma)\big(  \langle A_{0} B_{1} \rangle  + \langle A_{1} B_{0} \rangle - \langle A_{1} B_{1} \rangle \big), \label{eq:gammaIneq_1}
\end{equation}
where $c(\gamma) = 4\cos^{2} \left(\gamma + \frac{\pi}{6}\right) - 1$. Then we have the following:\\
\noindent (i) The local bound is given by $J_{\gamma}^{\mathrm{L}} = 12\cos^2\left( \gamma + \frac{\pi}{6} \right) - 4$.\\
\noindent (ii) The quantum bound is given by $J_{\gamma}^{\mathrm{Q}} = 8\cos^3\left( \gamma + \frac{\pi}{6}\right)$.\\
\noindent (iii) Up to local isometries there is a unique strategy that achieves $J_{\gamma} = J_{\gamma}^{\mathrm{Q}}$:
\begin{align}
    \rho_{Q_AQ_B}=\ket{\psi}\bra{\psi}\text{ where }\ket{\psi} = \frac{1}{\sqrt{2}}(\ket{00} + \ket{11}), \nonumber \\
    A_{0} = \sigma_{Z}, \ \ \ B_{0} = \sin 3\gamma \, \sigma_{Z} +\cos 3\gamma \, \sigma_{X}, \nonumber \\
    A_{1} = \cos\left(\frac{2\pi}{3}-2\gamma \right)  \sigma_{Z} + \sin\left(\frac{2\pi}{3}-2\gamma \right)  \sigma_{X}, \nonumber \\ B_{1} = \cos\left(\frac{\pi}{6}+\gamma \right)  \sigma_{Z} - \sin \left(\frac{\pi}{6}+\gamma \right)  \sigma_{X}.
    \label{eq:gammaStrat}
\end{align}
}
}
\fi

When $\gamma = 0$, this corresponds to the expression in \cref{eq:deltaIneq_1} for $\delta = \pi/6$, and when $\gamma = \pi/12$ we recover the CHSH expression.  The CHSH value for this family is given by $I_{\text{CHSH}} = \sin 3\gamma +3\cos(\gamma + \pi/6)$, and monotonically decreases in the interval $[3\sqrt{3}/2,2\sqrt{2}]$. The randomness certified by these self-tests is maximum for each CHSH value, summarized in our final proposition.

\ifarxiv
\begin{proposition}
\textit{The maximum randomness for strategies achieving a CHSH value in the range $s \in (2,3\sqrt{3}/2]$ is $\mathrm{2}$ bits, and is generated by the family of strategies in \cref{eq:deltaStrat_1}. For the range $s \in [3\sqrt{3}/2,2\sqrt{2}]$, the maximum is given by 
\begin{equation}
    1 + H_{\mathrm{bin}}\Big[\frac{1}{2} + \frac{s}{2}- \frac{3}{\sqrt{2}}\cos\Big(\frac{1}{3} \arccos \Big[-\frac{s}{2\sqrt{2}} \Big]\Big) \Big],
\end{equation}
where $H_{\mathrm{bin}}[\cdot]$ is the binary entropy, and is generated by the family of strategies in \cref{eq:gammaStrat}. }
\label{prep:maxRand}
\end{proposition}

\else{\medskip \noindent \textbf{Proposition 3.} \textit{The maximum randomness for strategies achieving a CHSH value in the range $s \in (2,3\sqrt{3}/2]$ is $\mathrm{2}$ bits, and is generated by the family of strategies in \cref{eq:deltaStrat_1}. For the range $s \in [3\sqrt{3}/2,2\sqrt{2}]$, the maximum is given by 
\begin{equation}
    1 + H_{\mathrm{bin}}\Big[\frac{1}{2} + \frac{s}{2}- \frac{3}{\sqrt{2}}\cos\Big(\frac{1}{3} \arccos \Big[-\frac{s}{2\sqrt{2}} \Big]\Big) \Big],
\end{equation}
where $H_{\mathrm{bin}}[\cdot]$ is the binary entropy, and is generated by the family of strategies in \cref{eq:gammaStrat}. }
\label{prep:maxRand}
}
\fi

Propositions 1--3 are proven in the \ifarxiv Appendices\else Supplemental Material~\cite{supp}\fi. In \cref{fig:maxRand}, we illustrate our results and compare them to a reliable lower bound on the minimum amount of randomness guaranteed by the same CHSH value~\cite{BrownDeviceIndependent2}. These two curves represent tight upper and lower bounds on the amount of DI randomness that can be certified by strategies achieving a particular CHSH value. 

\begin{figure}[h]
\includegraphics[width=8.4cm]{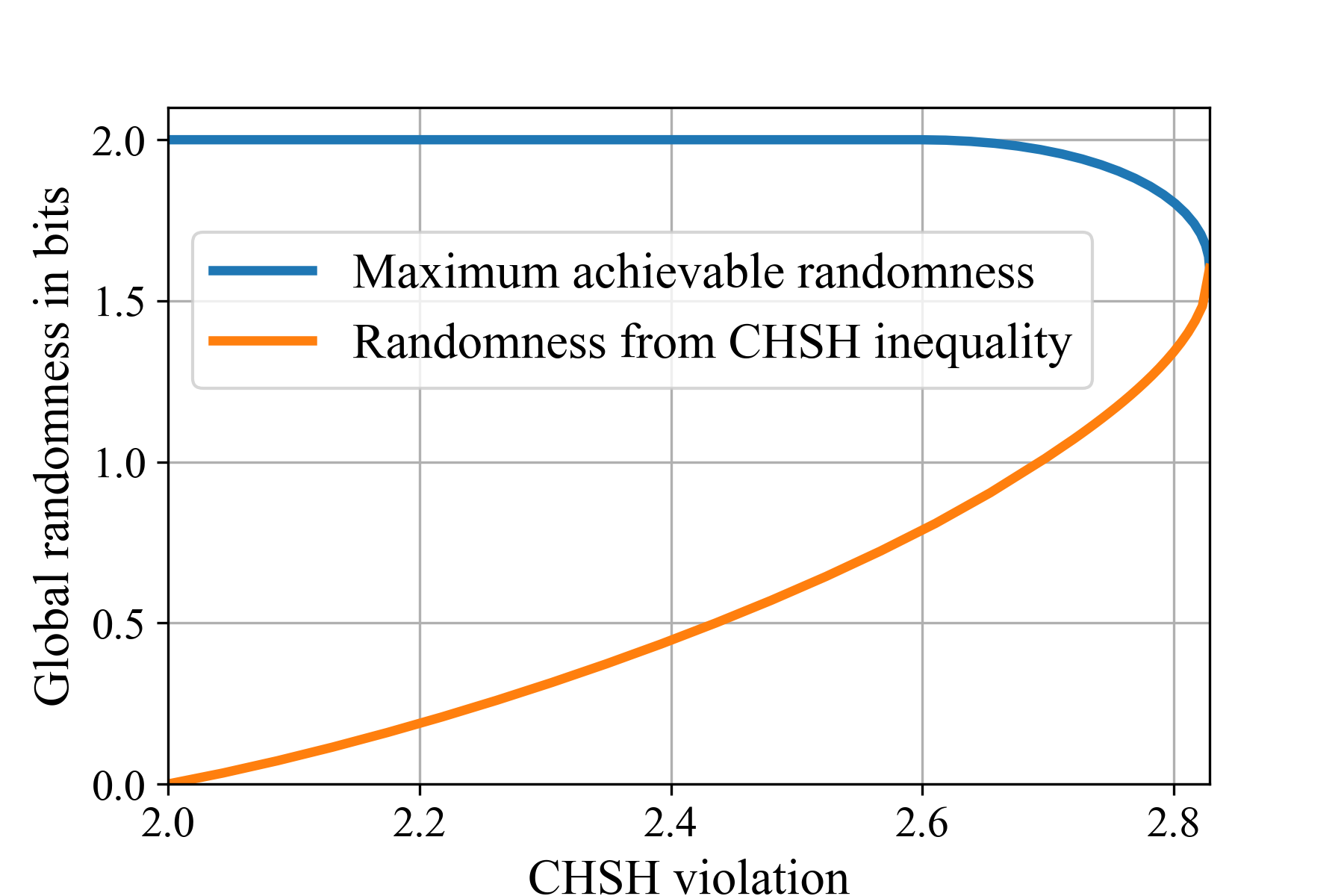}
\centering
\caption{The relationship between maximum randomness and CHSH value in the noiseless bipartite 2-input 2-output scenario. Plotted is the maximum achievable randomness for all quantum strategies that achieve a particular CHSH value (blue), and a reliable lower bound certified by the same CHSH value (orange) using analysis from~\cite{BrownDeviceIndependent2}. For the interval of values $(2,3\sqrt{3}/2]$ two bits of randomness are certified by the family of strategies in \cref{eq:deltaStrat_1}, and for the region $[3\sqrt{3}/2,2\sqrt{2}]$ the maximum is certified by the family of strategies in \cref{eq:gammaStrat}. Note that it is not the case that a CHSH value guarantees rates given by the blue curve; the blue curve gives a tight upper bound on achievable rates in a noiseless scenario.}
\label{fig:maxRand}
\end{figure}

In \cref{fig:alphaNoise} we explore the robustness of our constructions. We consider a Werner state noise model~\cite{Werner}, i.e., $\rho_{Q_{A}Q_{B}} = (1-p)\ketbra{\psi}{\psi} + p \ \mathbb{I}_{AB}/4$, where $p\in [0,1]$ is the weight of the uniform noise. For simplicity, we assume noiseless measurements. We use this state and measurements to simulate statistics from which reliable DI lower bounds can be generated using the techniques of~\cite{BrownDeviceIndependent2}. At each noise level, the randomness is optimized over the choice of self-test from \cref{eq:gammaIneq_1}. This is compared to the tilted CHSH expressions~\cite{AcinRandomnessNonlocality,BampsPironio}, where the tilting parameter is similarly optimized.   
\begin{figure}[h]
\includegraphics[width=8.4cm]{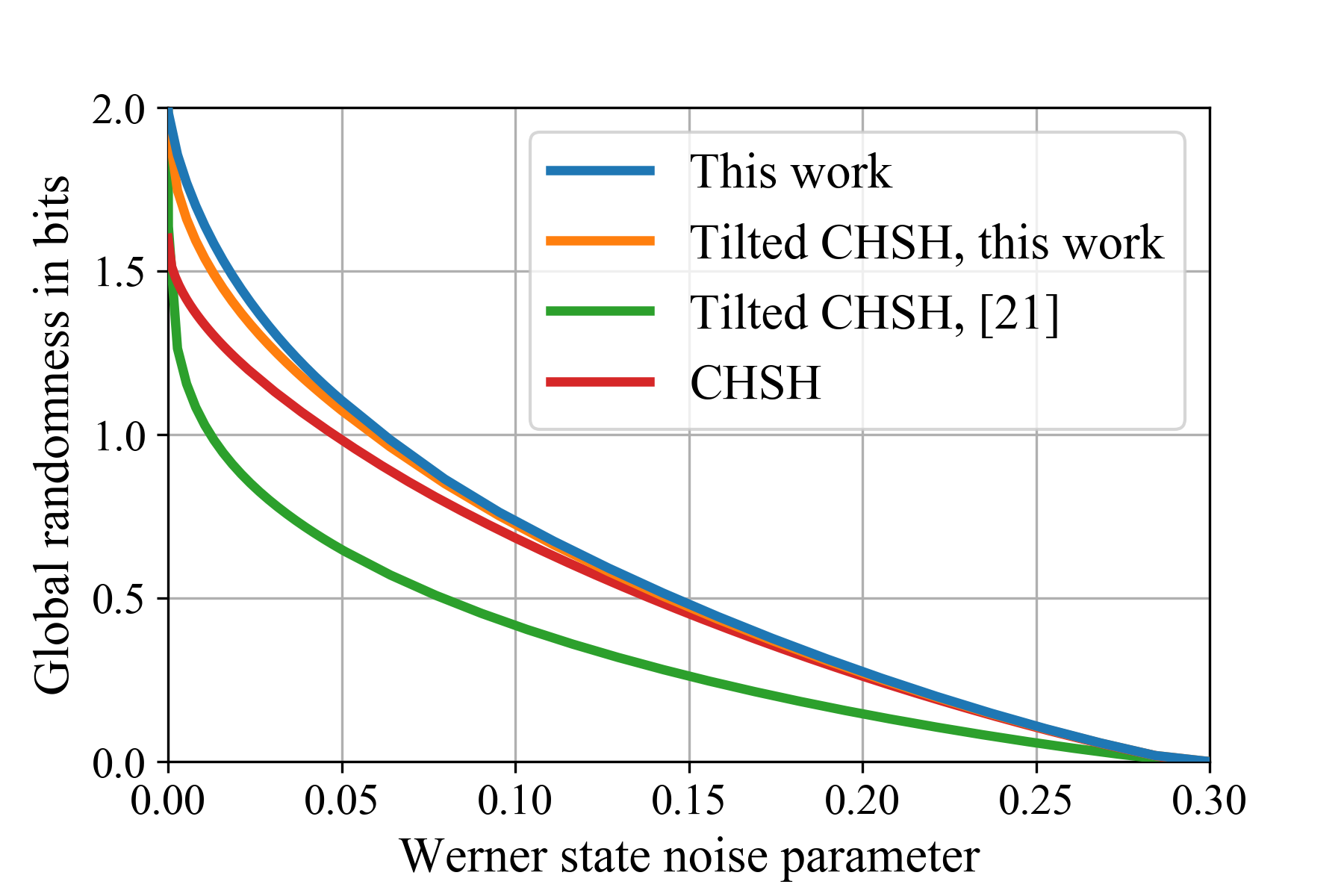}
\centering
\caption{Noise comparison for Bell expressions that certify maximum randomness in the bipartite 2-input 2-output scenario. Our constructions (blue) are compared to the tilted Bell inequalities from~\cite{AcinRandomnessNonlocality}, using both the numerical technique from~\cite{BrownDeviceIndependent2} (orange), and the min-entropy~\cite{AcinRandomnessNonlocality} (green). All these curves have been generated by optimizing the Bell expression (within the relevant families) at each value of the noise, and the new analysis shows improved rates of randomness generation over the CHSH statistic (red).}
\label{fig:alphaNoise}
\end{figure}

\ifarxiv\section{Discussion}\else{\medskip\noindent\it Discussion.|}\fi Our tight upper bound on the achievable DI randomness conditioned on the CHSH value shows that only when one approaches Tsirelson's bound does one need to sacrifice randomness for nonlocality. This comes from the fact that the optimal measurements needed to achieve $I_{\text{CHSH}} = 2\sqrt{2}$ have correlated outcomes, whereas correlations that satisfy $I_{\text{CHSH}} \leq 3\sqrt{3}/2$ can have uniform measurement outcomes. When there is zero noise, such a distribution can be used to generate 2 bits of randomness (using the $\gamma=0$ strategy in \cref{eq:gammaStrat}). As noise is added, using the family of Bell inequalities that self-test a distribution with a CHSH value greater than $3\sqrt{3}/2$ (obtained by increasing $\gamma$), we can continue to certify more randomness than would be possible using CHSH inequality at that noise level. Taking the optimal value of $\gamma$ at each level of noise we find that the Bell expressions tend to the CHSH statistic as the noise approaches the boundary where no randomness can be certified. In this sense, CHSH is the most robust statistic, which is natural since it defines a facet of the local polytope and so becomes the only Bell inequality that can be violated with high enough noise. We also remark that the robustness of the tilted CHSH inequalities presented here is higher than that of~\cite{AcinRandomnessNonlocality} (cf.\ the orange vs green curves in Fig.~\ref{fig:alphaNoise}). This is a result of using improved numerical techniques to bound the conditional von Neumann entropy directly rather than the min-entropy that is used in~\cite{AcinRandomnessNonlocality}.  

Based on an experimental estimate of the noise, a Bell inequality from one of our families could be chosen that maximises the certifiable randomness (along the lines of \cref{fig:alphaNoise}). Knowledge of the full distribution could also boost the noise performance or be used to search for improved protocols in the presence of noise.  However, the use of more parameters would lead to a penalty when finite size effects are accounted for. We leave the question of how our construction performs in other noise models, such as detector efficiency, to future work, and pose an open question as to if our construction is truly optimal in the noisy regime.          

One other potential application of our constructions is to blind randomness expansion~\cite{MillerBlind,HonghaoBlind,MetgerGEAT}, where Alice tries to certify local randomness from one device without trusting the other. Since their outputs are uncorrelated following a self-test from Proposition \ifarxiv\ref{prep:delta}\else 1\fi, such a statistic could be used to generate the optimal 1 bit of local randomness in the blind setting.

Finally, it would be interesting to further investigate analogous results in multi-partite scenarios~\cite{WoodheadMermin,GrasselliMulti} or those with more inputs or outputs~\cite{SarkarSelfTest}. Indeed,~\cite{GrasselliMulti} showed maximal randomness for the outputs of two parties based on a three party Mermin-Ardehali-Belinskii-Klyshko inequality~\cite{CollinesMulti}, and with advancements in computational efficiency from new numerical techniques~\cite{BrownDeviceIndependent2,MasiniDI} alongside self-testing results~\cite{McKagueGraph,WuWselfTest}, multi-partite DI-RG has many avenues to explore.

\ifarxiv\acknowledgements\else\medskip\noindent{\it Acknowledgements}|\fi This work was supported by EPSRC via the Quantum Communications Hub (Grant No.\ EP/T001011/1) and Grant No.\ EP/SO23607/1.

%


\onecolumngrid
\appendix

\section{Proof of Propositions \ifarxiv\ref{prep:delta} and \ref{prep:gamma}\else 1 and 2\fi}
The goal of this section is to prove Propositions \ifarxiv\ref{prep:delta} and \ref{prep:gamma}\else 1 and 2\fi, which are restated below:

\vspace{0.5cm}

\noindent \textbf{Proposition 1} ($I_{\delta}$-family of self-tests). \textit{Let $0 < \delta \leq \pi/6$, and define the family of Bell expressions parameterized by $\delta$,
\begin{equation}
    I_{\delta} = \langle A_{0} B_{0} \rangle + \frac{1}{\sin \delta}\Big(  \langle A_{0} B_{1} \rangle  + \langle A_{1} B_{0} \rangle \Big)   
    - \frac{1}{\cos 2\delta} \langle A_{1} B_{1} \rangle . \label{eq:deltaIneq_supp}
\end{equation}
Then we have the following:
\begin{enumerate}[(i)]
    \item The local bound is given by $I_{\delta}^{\mathrm{L}} = -1 + \frac{2}{\sin{\delta}} + \frac{1}{\cos{2\delta}}$.
    \item The quantum bound is given by $ I_{\delta}^{\mathrm{Q}} = \frac{2\cos^{3}\delta}{\cos 2\delta \sin \delta}$.
    \item Up to local isometries there is a unique strategy that achieves $I_{\delta} = I_{\delta}^{\mathrm{Q}}$:
\begin{gather}
    \rho_{Q_AQ_B}=\ket{\psi}\bra{\psi}\text{ where }\ket{\psi} = \frac{1}{\sqrt{2}}(\ket{00} + \ket{11}), \nonumber \\
    A_{0} = \sigma_{Z}, \ \ \ B_{0} = \sigma_{X}, \nonumber \\
    A_{1} = -\sin \delta \, \sigma_{Z} + \cos \delta \, \sigma_{X}, \nonumber \\ B_{1} = \cos \delta \, \sigma_{Z} - \sin \delta \, \sigma_{X}.
    \label{eq:deltaStrat_supp}
\end{gather}
\end{enumerate}
}

\vspace{0.5cm}

\noindent \textbf{Proposition 2} ($J_{\gamma}$-family of self-tests). \textit{Let $0 \leq \gamma \leq \pi/12$, and define the family of Bell expressions parameterized by $\gamma$,
\begin{equation}
    J_{\gamma} = \langle A_{0} B_{0} \rangle  + \Big(4\cos^{2} \Big(\gamma + \frac{\pi}{6}\Big) - 1\Big)\Big(  \langle A_{0} B_{1} \rangle  + \langle A_{1} B_{0} \rangle - \langle A_{1} B_{1} \rangle \Big). \label{eq:gammaIneq_supp}
\end{equation}
Then we have the following:
\begin{enumerate}[(i)]
    \item The local bound is given by $J_{\gamma}^{\mathrm{L}} = 12\cos^2\Big( \gamma + \frac{\pi}{6} \Big) - 4$.
    \item The quantum bound is given by $J_{\gamma}^{\mathrm{Q}} = 8\cos^3\Big( \gamma + \frac{\pi}{6}\Big)$.
    \item Up to local isometries there is a unique strategy that achieves $J_{\gamma} = J_{\gamma}^{\mathrm{Q}}$:
\begin{gather}
    \rho_{Q_AQ_B}=\ket{\psi}\bra{\psi}\text{ where }\ket{\psi} = \frac{1}{\sqrt{2}}(\ket{00} + \ket{11}), \nonumber \\
    A_{0} = \sigma_{Z}, \ \ \ B_{0} = \sin 3\gamma \, \sigma_{Z} +\cos 3\gamma \, \sigma_{X}, \nonumber \\
    A_{1} = \cos\Big(\frac{2\pi}{3}-2\gamma \Big)  \sigma_{Z} + \sin\Big(\frac{2\pi}{3}-2\gamma \Big) \sigma_{X}, \nonumber \\ B_{1} = \cos\Big(\frac{\pi}{6}+\gamma \Big) \sigma_{Z} - \sin \Big(\frac{\pi}{6}+\gamma \Big) \sigma_{X}.
    \label{eq:gammaStrat_supp}
\end{gather}
\end{enumerate}
}

\vspace{0.5cm}

We follow the same method for both cases. For part $(i)$, the local bound can be found by setting the observables $A_{x},B_{y}$ to $\pm1$, corresponding to an extremal or deterministic strategy. Since these are the vertices of the local polytope, one such combination will be the optimal local strategy. 

For the quantum bound in part $(ii)$, a sum-of-squares (SOS) decomposition is found for the Bell expression offset by its claimed maximum quantum value, exploiting the symmetry of the Bell expression under relabelling of $A$ and $B$~\cite{BampsPironio}. The existence of an SOS decomposition proves the maximum quantum value claimed, and is detailed in \cref{sec:SOS1,sec:SOS2}. 

For the self-test in part $(iii)$ we use the resulting SOS decomposition in combination with Jordan's lemma~\cite{Jordan}. This simplifies the analysis to qubits, and we derive a system of non-linear equations satisfied by any state and measurements that achieve the maximum quantum value. The resulting system is then analytically solved, and we show the only state and measurements for which these equations are satisfied is given by the target strategy (\cref{eq:deltaStrat_supp,eq:gammaStrat_supp}) up to local unitaries. This process is detailed in \cref{sec:jordan,sec:delta,sec:gamma}, and completes the proof of Propositions \ifarxiv\ref{prep:delta} and \ref{prep:gamma}\else 1 and 2\fi. We remark that our self-tests define hyperplanes tangential to the corresponding strategy on the boundary of the quantum set~\cite{GohGeometry}.  

For completeness, in \cref{sec:convex}, we show from Jordan's lemma that the private randomness of any strategy that saturates the quantum bounds in Propositions \ifarxiv\ref{prep:delta} and \ref{prep:gamma} \else 1 and 2 \fi is equal to that of the target strategy. 

\subsection{Self-testing and sum-of-squares decompositions}\label{sec:SOS1}

We consider only the exact self-testing statement in this work, and leave proof of the robust statement for future work. We begin by defining self-testing.

\begin{definition}[Self-test]
Let the observables $A_{x}$, $B_{y}$ and pure state $\ket{\psi}_{Q_{A}Q_{B}}$ be the target two-qubit strategy, and let $S$ be a Bell operator. The inequality $\langle S \rangle \leq I^{\text{Q}}$ \emph{self-tests} the target state and measurements if for all physical quantum strategies $(\tilde{\rho}_{\tilde{Q}_{A}\tilde{Q}_{B}},\tilde{A}_{x},\tilde{B}_{y})$ that satisfy $\langle \tilde{S} \rangle = I^{\text{Q}}$, there exists a local isometry $V:\mathcal{H}_{\tilde{Q}_{A}}\otimes \mathcal{H}_{\tilde{Q}_{B}}\otimes \mathcal{H}_{E} \rightarrow \mathcal{H}_{Q_{A}}\otimes \mathcal{H}_{Q_{B}} \otimes \mathcal{H}_{\mathrm{Junk}}$, $V = V_{A}\otimes V_{B} \otimes \mathbb{I}_{E}$, and ancillary state $\ket{\xi}_{\text{Junk}}$ such that, for the purification $\ket{\Psi}_{\tilde{Q}_{A}\tilde{Q}_{B}E}$ of $\tilde{\rho}_{\tilde{Q}_{A}\tilde{Q}_{B}}$,
\begin{equation}
    V\Big[ (\tilde{A}_{x} \otimes \tilde{B}_{y} \otimes \mathbb{I}_{E}) \ket{\Psi}_{\tilde{Q}_{A}\tilde{Q}_{B}E} \Big] 
    =   (A_{x} \otimes B_{y}) \ket{\psi}_{Q_{A}Q_{B}} \otimes \ket{\xi}_{\text{Junk}}. 
\end{equation}
\label{def:selfTest}
\end{definition}
\noindent Throughout this appendix, we refer to the physical state and measurements we are trying to self-test as the ``reference'', denoted with a tilde. The strategies in \cref{eq:deltaStrat_supp,eq:gammaStrat_supp} are then the ``target'' state and measurements; which target strategy we refer to will be clear from the context. 

For a Bell operator $S$ that defines the quantum Bell inequality $\langle S \rangle \leq I^{\text{Q}}$, the operator $\bar{S}:= I^{\text{Q}}\mathbb{I} - S$, satisfies $\bra{\phi}\bar{S}\ket{\phi}\geq 0$ for all quantum states $\ket{\phi}$, i.e., $\bar{S} \succeq 0$. If there exists a set of operators $P_{i}$ that are polynomials of $A_{x},B_{y}$ and satisfy 
\begin{equation}
    \bar{S} = \sum_{i} P_{i}^{\dagger}P_{i},
\end{equation}
then we have found a sum-of-squares (SOS) decomposition of the operator $\bar{S}$: positivity of $\bar{S}$ follows directly from the fact that $K^{\dagger}K\succeq 0$ for any operator $K$.

SOS decompositions can be used to enforce algebraic constraints on any state and measurements that satisfy $\langle S \rangle = I^{\text{Q}}$, since this implies
\begin{equation}
    \langle \bar{S} \rangle  = \sum_{i} \bra{\psi}P_{i}^{\dagger}P_{i}\ket{\psi} = 0\,.
\end{equation}
This can only hold if $P_{i}\ket{\psi} = 0$ for all $i$. Relations of this form are used to prove the self-testing statement in \cref{def:selfTest}.

\subsection{SOS decomposition for the inequalities in Propositions \ifarxiv\ref{prep:delta} and \ref{prep:gamma}\else 1 and 2\fi}
Finding an SOS decomposition can be recast as a semidefinite program (SDP)~\cite{BampsPironio}. We start by considering a vector $\bm{R}=[R_{0},...,R_{k},...]^{\text{T}}$ whose components are linear combinations of $A_0$, $A_1$, $B_0$ and $B_1$. We consider the case where each polynomial $P_i$ is linear, writing $P_{i} = \sum_{k} q_{i}^{k}R_{k}$ for some coefficients $\{q_i^k\}_k$. Then
\begin{align}
    \bar{S} &= \sum_{i} P_{i}^{\dagger}P_{i} \nonumber \\
    &= \sum_{kj}   R_{k}^{\dagger}\left(\sum_{i}\big(q_{i}^{k}\big)^{*} q_{i}^{j}\right)R_{j} \nonumber \\
    &= \sum_{kj}R_{k}^{\dagger}M_{kj}R_{j} = \bm{R}^{\dagger}M\bm{R}\,, \label{eq:SOS}
\end{align}
where $M$ is the Gram matrix of the set of vectors $\{\bm{q}^{k}\}$. Since $M$ is a Gram matrix, it is positive semidefinite by construction.  We can hence use semidefinite programming to find an $M \succeq 0$ that satisfies \cref{eq:SOS}, and then find the polynomials $P_{i}$ via the matrix square root:
\begin{align}
    \bar{S} &= \bm{R}^{\dagger}M\bm{R} = \Big( \sqrt{M}\bm{R} \Big)^{\dagger} \Big( \sqrt{M} \bm{R} \Big). \label{eq:sosDecomp}
\end{align}
Since each entry of the vector $\sqrt{M}\bm{R}$ takes the form $[\sqrt{M}\bm{R}]_{i} = \sum_{k}[\sqrt{M}]_{ik}R_{k}$, we find that $P_{i} =[\sqrt{M}\bm{R}]_{i}$ provides the set of polynomials that satisfies \cref{eq:SOS}.

For the Bell operator $\bar{S}_{\delta} = I_{\delta}^{\text{Q}}\mathbb{I} - I_\delta$, where
\begin{equation}
     I_\delta = A_{0} B_{0}  + \frac{1}{\sin\delta}\Big(   A_{0} B_{1}   + \ A_{1} B_{0}  \Big)   
    - \frac{1}{\cos 2\delta}  A_{1} B_{1} , \label{eq:deltaOp}
\end{equation}
the SOS decomposition is given by the following lemma.
\begin{lemma}[$I_{\delta}$-family SOS decomposition]
Let $\bm{R} = [R_{0},R_{1},R_{2},R_{3}]^{\text{T}}$, where 
\begin{align}
    R_{0} = \frac{1}{\sqrt{2}}(B_{1}  - A_{ 1}), \\
    R_{1} = \frac{1}{\sqrt{2}}(B_{0} -  A_{0}), \\
    R_{2} = \frac{1}{\sqrt{2}}(B_{1} +  A_{1}), \\
    R_{3} = \frac{1}{\sqrt{2}}(B_{0} + A_{0}).
\end{align}
For every $\delta\in(0,\pi/6]$, the Bell expressions $\bar{S}_{\delta}$ can be written as a SOS decomposition $\bar{S}_{\delta}=\bm{R}^{\dagger}M_\delta\bm{R}$ where
\begin{equation}
    M_{\delta} = 
    \begin{bmatrix}
        (\alpha - 1/2)\zeta  & \beta &  0 & 0 \\
        \beta & \alpha + 1/2 & 0 & 0 \\
        0 & 0 & (\alpha + 1/2)\zeta & -\beta \\
        0 & 0 & -\beta & \alpha - 1/2
    \end{bmatrix},
\end{equation}
for $\alpha = 1/(2\tan \delta)$, $\beta = 1/(2\sin \delta)$ and $\zeta = 1/\cos 2\delta$.

The maximum quantum value of $I_\delta$ is $I_{\delta}^{\mathrm{Q}} = \frac{2\cos^{3}\delta}{\cos 2\delta \sin \delta}$.
\label{lem:SOS_delta}
\end{lemma}

\begin{proof}
 The claim $\bar{S}_{\delta}=\bm{R}^{\dagger}M_\delta\bm{R}$ can be verified by direct calculation.  Since $\bar{S}_{\delta}\succeq 0$ we have $\langle I_\delta\rangle\leq I_{\delta}^{\mathrm{Q}}$, but the quantum strategy given in~\eqref{eq:deltaStrat_supp} shows that this bound is achievable.
\end{proof}

Four polynomials $P_{i}(\delta)$ emerge from this decomposition:
\begin{align}
    P_{0}(\delta) &= k_{+}\big[ R_{0} + (\sin \delta + \cos \delta)R_{1} \big], \\
    P_{1}(\delta) &= (\sin \delta + \cos \delta)P_{0}(\delta), \\
    P_{2}(\delta) &= k_{-}\big[ R_{2} + (\sin \delta - \cos \delta)R_{3} \big], \\
    P_{3}(\delta) &= (\sin \delta - \cos \delta)P_{2}(\delta),
\end{align}
where 
\begin{align}
    k_{\pm} &= \frac{1}{\sqrt{2\sin\delta\,(\cos\delta\pm\sin\delta)(2\pm\sin(2\delta))}}\,.
\end{align}

Similarly, the shifted Bell operator for the $J_{\gamma}$ family is given by $\bar{S}'_{\gamma} = J_{\gamma}^{\text{Q}}\mathbb{I} - J_{\gamma}$, where
\begin{equation}
     J_{\gamma} = A_{0} B_{0}  + (4\cos^2(\gamma + \pi/6)-1)\Big(   A_{0} B_{1}   +  A_{1} B_{0}  -  A_{1} B_{1} \Big), \label{eq:gammaOp_s}
\end{equation}
and we have the following SOS decomposition:
\begin{lemma}[$J_{\gamma}$-family SOS decomposition]
Let $\bm{R}$ be as defined in \cref{lem:SOS_delta}. 
For every $\gamma\in[0,\pi/12]$, the Bell expressions $\bar{S}'_{\gamma}$ can be written as a SOS decomposition $\bar{S}'_{\gamma}=\bm{R}^{\dagger}M'_\gamma\bm{R}$ where
\begin{equation}
    M'_{\gamma} = 
    \begin{bmatrix}
        1/2 + \mu(4\mu^2 - 2\mu -1)  &  2\mu^2 - 1/2 &  0 & 0 \\
        2\mu^2 - 1/2 & \mu + 1/2 & 0 & 0 \\
        0 & 0 &  -1/2 + \mu(4\mu^2 + 2\mu -1) & 1/2 -2\mu^2 \\
        0 & 0 & 1/2 -2\mu^2 & \mu - 1/2
    \end{bmatrix},
\end{equation}
where $\mu = \cos( \gamma + \pi/6)$.

The maximum quantum value of $J_\gamma$ is $J_{\gamma}^{\mathrm{Q}} = 8\mu^3$.
\label{lem:SOS_gamma}
\end{lemma}

This can be proven in exactly the same way as \cref{lem:SOS_delta} and gives rise to the polynomials
\begin{align}
    P'_{0}(\gamma) &= c_{+}\big[ (2\mu-1)R_{0} + R_{1} \big], \\
    P'_{1}(\gamma) &= \frac{P'_{0}(\gamma)}{2\mu-1}, \\
    P'_{2}(\gamma) &= c_{-}\big[ (2\mu+1)R_{2} - R_{3} \big], \\
    P'_{3}(\gamma) &= - \frac{P'_{2}(\gamma)}{2\mu + 1},
\end{align}
where 
\begin{align}
    c_{\pm} &= \frac{4\mu^2-1}{2\sqrt{4\mu^3 \mp 2\mu^2 \pm 1} }.
\end{align}

Lemmas~\ref{lem:SOS_delta} and~\ref{lem:SOS_gamma} establish the quantum bounds in Propositions \ifarxiv\ref{prep:delta} and \ref{prep:gamma}\else 1 and 2\fi, and give us the tools needed to prove the self-testing claims, following a reduction to qubits detailed in the next section.  
\label{sec:SOS2}

\subsection{Applying Jordan's lemma}
We can use the polynomials derived above for both families of inequalities to impose algebraic constraints on the state and measurements that satisfy $\langle \bar{S}\rangle = 0$.  We employ Jordan's lemma~\cite{Jordan}, a unique simplification that can be made in the 2-input 2-output scenario~\cite{PABGMS,BhavsarDI}. The lemma states that for two observables $A_{0}$ and $A_{1}$ on a Hilbert space $\mathcal{H}$, each with eigenvalues $\pm 1$, there exists a basis transformation such that both are simultaneously block diagonal with block size no greater than two. The Hilbert space decomposes into this block diagonal structure, and, by dilating where necessary, we can take each block to be a qubit system. There then exists a block diagonal density operator that reproduces the statistics of the original system. This can be summarised as follows.
\begin{lemma}[Jordan's lemma]
Let $A_{0}$ and $A_{1}$ be two binary observables on a Hilbert space $\mathcal{H}_{A}$. Then there exists a basis in which $A_0$ and $A_1$ are block diagonal with block dimensions at most $2$. Moreover, for every state and set of measurements on $\mathcal{H}_{\tilde{Q}_{A}}\otimes \mathcal{H}_{\tilde{Q}_{B}}$ that generates a post-measurement state $\rho_{ABE}$, there exists another state and set of measurements, given by a convex combinations of two-qubit systems, that generates the same post-measurement state.  
\end{lemma}
\noindent For self-testing literature that also utilises Jordan's lemma, see e.g.~\cite{BardynSelfTest,SekatskiBuilidngBlocks,ValcarceSelfTest}.

It is known that a full reduction to a convex combination of two-qubit strategies with measurements in the $XZ$-plane is sufficient for evaluating the global entropy~\cite{PABGMS,BhavsarDI}. By employing Jordan's lemma to systems $\tilde{Q}_{A}$ and $\tilde{Q}_{B}$, the resulting parameterization of a single two-qubit strategy is given by 7 parameters: 3 for the state, which can be taken to be diagonal in the Bell basis, and 4 for the measurements, one defining each angle in the $XZ$-plane. Let
\begin{align}
    \ket{\Phi_{0}}&= \frac{1}{\sqrt{2}}(\ket{00} + \ket{11}), \nonumber \\
    \ket{\Phi_{1}}&= \frac{1}{\sqrt{2}}(\ket{00} - \ket{11}), \nonumber \\
    \ket{\Phi_{2}}&= \frac{1}{\sqrt{2}}(\ket{01} + \ket{10}), \nonumber \\
    \ket{\Phi_{3}}&= \frac{1}{\sqrt{2}}(\ket{01} - \ket{10}).
\end{align}
The two-qubit state is given by
\begin{equation}\label{eq:Bellbasis}
    \rho = \sum_{\alpha=0}^3\lambda_{\alpha}\ketbra{\Phi_{\alpha}}{\Phi_{\alpha}},
\end{equation}
where $\lambda_{\alpha} \geq 0$ and $\sum_\alpha\lambda_\alpha = 1$. The measurements are given by
\begin{align}
    A_{x} = \cos a_{x} \, \sigma_{Z} + \sin a_{x} \, \sigma_{X}, \nonumber \\
    B_{y} = \cos b_{y} \, \sigma_{Z} + \sin b_{y} \, \sigma_{X},\label{eq:mmts}
\end{align}
where $-\pi < a_{x},b_{y} \leq \pi$, $x,y\in \{0,1\}$. 
See~\cite{PABGMS,BhavsarDI} for details of this reduction.

Our methodology will be to show that the only two-qubit strategy that satisfies the relations imposed by the SOS polynomials is the target strategy up to local unitaries, hence the extraction map can be written in terms of unitaries that rotate each Jordan block to the target. \label{sec:jordan}

\subsection{Proof of the self-testing claim for the $I_{\delta}$-family}
We now prove the self-testing claim in Proposition \ifarxiv\ref{prep:delta}\else 1\fi, i.e., that the family of inequalities in \cref{eq:deltaIneq_supp} self-tests the state and measurements in \cref{eq:deltaStrat_supp}.
\begin{theorem}[Self-testing the $I_{\delta}$-family]
The family of Bell expressions in \cref{eq:deltaIneq_supp} provides a self-test for the two-qubit state and family of measurements in \cref{eq:deltaStrat_supp} according to \cref{def:selfTest}.  [Equivalently, up to local isometries, the only state and measurements that satisfy $I_\delta=I_\delta^{\mathrm{Q}}$ are those of \cref{eq:deltaStrat_supp}.]
\label{thm:selfTest}
\end{theorem}

\begin{proof}
The previous section implies that it is sufficient to consider two qubit states that are diagonal in the Bell basis as in~\eqref{eq:Bellbasis} and measurements of the form~\eqref{eq:mmts}. Consider the expectation value of the operator $\bar{S}_{\delta}$ for a two qubit state $\rho$ and measurements $A_{x},B_{y}$ that saturate the inequality in \cref{eq:deltaIneq_supp}:
\begin{align}
    \langle \bar{S}_{\delta}\rangle &= \sum_{i} \langle P_{i}^\dagger(\delta)P_{i}(\delta) \rangle \nonumber \\
    &= \sum_{i} \text{Tr}[\rho P_{i}^\dagger(\delta)P_{i}(\delta)] \nonumber \\
    &= \sum_{i} \sum_{\alpha} \lambda_{\alpha}\|P_{i}(\delta)\ket{\Phi_{\alpha}}\|^{2} = 0.
\end{align}
Since $\lambda_{\alpha}\geq 0, \ \|P_{i}(\delta)\ket{\Phi_{\alpha}}\|^{2} \geq 0$, we have
\begin{equation}
    \lambda_{\alpha}\|P_{i}(\delta)\ket{\Phi_{\alpha}}\|^{2} = 0 \ \forall i \ \forall \alpha. \label{eq:SOScond}
\end{equation}
Without loss of generality, suppose $\lambda_{0} \neq 0$ (if $\lambda_{0} = 0$, then for some $\alpha'$ with $\lambda_{\alpha'} \neq 0$, there is a local unitary $U$ such that $U\otimes \mathbb{I}\ket{\Phi_{\alpha'}} = \ket{\Phi_{0}}$, so cases where $\lambda_{0} = 0$ are equivalent to the case $\lambda_{0} \neq 0$ up to local unitaries).  Note that $(U\otimes U^\mathrm{T})\ket{\Phi_0}=\ket{\Phi_0}$ for all single qubit unitaries $U$, where $U^{\mathrm{T}}$ is the transpose of $U$ in the $\{\ket{0},\ket{1}\}$ basis. It follows that we can take $a_0=0$, i.e., $A_0=\sigma_Z$.

By~\eqref{eq:SOScond}, we have that $P_{i}(\delta)\ket{\Phi_{0}} = 0$ for $i=0,2$ (the cases $i=1,3$ are identical by linear dependence). Using the form of the measurements (cf.~\eqref{eq:mmts}), we arrive at the system of nonlinear equations
\begin{align}
    (\sin \delta + \cos \delta)\sin b_0+\left( \sin b_{1} - \sin a_{1}  \right) = 0, \label{eq:1} \\
    (\sin \delta - \cos \delta)\sin b_0+\left( \sin b_{1} + \sin a_{1}  \right) = 0, \label{eq:2} \\
    (\sin \delta + \cos \delta)\left(\cos b_{0}-1\right) + \left( \cos b_{1} - \cos a_{1} \right) = 0, \label{eq:3} \\
    (\sin \delta - \cos \delta)\left( \cos b_{0} +1 \right) + \left( \cos b_{1} + \cos a_{1} \right) = 0. \label{eq:4} 
\end{align}
Subtracting \cref{eq:2} from \cref{eq:1}, and \cref{eq:4} from \cref{eq:3} gives
\begin{align}
    \sin a_{1} = \sin b_{0} \cos \delta, \label{eq:13}\\
    \cos a_{1} = \cos b_{0} \cos \delta - \sin \delta\,. \label{eq:14}
\end{align}
Then using $\sin^2a_1+\cos^2a_1=1$ we recover
\begin{align}
    1 &= \sin^2 b_0 \cos^2 \delta+\left(\cos b_0 \cos \delta - \sin \delta \right)^2 \nonumber \\
    &= 1- \sin 2\delta\, \cos b_0, 
   \label{eq:15}
\end{align}
and hence we have $\cos b_0=0$. Since $-\pi<b_0\leq\pi$ we have $b_0=\pm\pi/2$, i.e., $B_0=\pm\sigma_X$.  Noting that $\sigma_Z\otimes\sigma_Z$ has no effect on $\ket{\Psi_0}$ and that $\sigma_Z\sigma_X\sigma_Z=-\sigma_X$, we can take $b_0=\pi/2$, i.e., $B_0=\sigma_X$ without loss of generality. Then, $\sin b_0=1$.

Using these in~\eqref{eq:1}--\eqref{eq:4} we find $\sin b_1=-\sin\delta$ and $\cos b_1=\cos\delta$, hence $b_1=-\delta$.  Similarly, $\sin a_1=\cos\delta$ and $\cos a_1=-\sin\delta$ so we have $A_1=-\sin\delta\,\sigma_Z+\cos\delta\,\sigma_X$ and $B_1=\cos\delta\,\sigma_Z-\sin\delta\,\sigma_X$, recovering the observables in \cref{eq:deltaStrat_supp}. We have therefore proved the self-testing of the measurements.

For the state, consider $\|P_i(\delta)\ket{\Phi_{\alpha}}\|^2$ for $i=0,2$ and $\alpha=1,2,3$. By direct calculation, using the observables we found above, we find all of these to be $\cos^2\delta$. Hence, by~\eqref{eq:SOScond}, we must have $\lambda_1=\lambda_2=\lambda_3=0$ and thus $\lambda_0=1$.

Finally we derive the extraction map from \cref{def:selfTest}. According to Jordan's lemma, both Hilbert spaces decomposes block-diagonally with $2\times 2$ blocks. This is equivalent to identifying $\mathcal{H}_{\tilde{Q}_{A}} = \mathcal{H}_{F_{A}}\otimes \mathcal{H}_{Q_{A}}$ where $F_{A}$ is a system that flags the $2\times 2$ Jordan block, and $Q_{A}$ is a qubit system (similarly for $\mathcal{H}_{\tilde{Q}_{B}}$). With purifying system $E$, the purified state hence takes the form
\begin{equation}
    \ket{\Psi}_{\tilde{Q}_{A}\tilde{Q}_{B}E} = \sum_{ij} \sqrt{p_{ij}} \ket{ij}_{F_{A}F_{B}} \otimes \ket{\varphi_{ij}}_{Q_{A}Q_{B}} \otimes \ket{ij}_{E},  
\end{equation}
where $\rho_{\tilde{Q}_{A}\tilde{Q}_{B}} = \text{Tr}_{E}\Big[\ketbra{\Psi}{\Psi}_{\tilde{Q}_{A}\tilde{Q}_{B}E}\Big] =  \sum_{ij} p_{ij} \ketbra{ij}{ij}_{F_{A}F_{B}} \otimes \ketbra{\varphi_{ij}}{\varphi_{ij}}_{Q_{A}Q_{B}}$ is the state shared by the devices. Similarly, the measurements admit the decomposition 
\begin{equation}
    \tilde{A}_{x} \otimes \tilde{B}_{y} = \sum_{ij} \ketbra{ij}{ij}_{F_{A}F_{B}} \otimes A_{x}^i \otimes B_{y}^j. 
\end{equation}

Above we established that, up to local unitaries, the only two qubit strategy that can achieve $\langle \tilde{S}_{\delta} \rangle = I^{\text{Q}}_{\delta}$ is the target in \cref{eq:deltaStrat_supp}. Therefore, for every measurement pair $A^i_{x} \otimes B^j_{y}$ and state $\ket{\varphi_{ij}}$ there exist local unitaries $U_A^i:\cH_{Q_A}\to\cH_{Q_A}$ and $U_B^j:\cH_{Q_B}\to\cH_{Q_B}$ such that $U_A^iA_x^i(U_A^i)^{\dagger} = A_x$, $U_B^jB_y^j(U_B^j)^{\dagger} = B_y$, and $(U_A^i\otimes U_B^j)\ket{\varphi_{ij}} = \ket{\Phi_{0}}$. Thus, if we define the unitary 
\begin{equation}
    V = \sum_{ij}\ketbra{ij}{ij}_{F_{A}F_{B}} \otimes U_A^i\otimes U_B^j \otimes \mathbb{I}_{E}\,,
\end{equation}
then we have the extraction
\begin{align}
    V (\tilde{A}_{x} \otimes \tilde{B}_{y} \otimes \mathbb{I}_{E}) V^{\dagger}V \ket{\Psi}_{\tilde{Q}_{A}\tilde{Q}_{B}E} 
    =   (A_{x} \otimes B_{y}) \ket{\Phi_{0}}_{Q_{A}Q_{B}} \otimes \Bigg( \sum_{ij} \sqrt{p_{ij}}\ket{ij}_{F_{A}F_{B}} \otimes \ket{ij}_{E}\Bigg).
\end{align}
This is of the form in \cref{def:selfTest}, and completes the self-testing proof.
\end{proof}
\label{sec:delta}

\subsection{Proof of the self-testing claim for the $J_{\gamma}$-family}
We follow an identical methodology to the previous section to prove the self-testing claim in Proposition \ifarxiv\ref{prep:gamma}\else 2\fi.

\begin{theorem}[Self-testing the $J_{\gamma}$-family]
The family of Bell expressions in \cref{eq:gammaIneq_supp} provides a self-test for the family of two qubit states and measurements in \cref{eq:gammaStrat_supp} according to \cref{def:selfTest}. [Equivalently, up to local isometries, the only state and measurements that satisfy $J_\gamma=J_\gamma^{\mathrm{Q}}$ are those of \cref{eq:gammaStrat_supp}.]
\label{thm:gammaSelfTest}
\end{theorem}

\begin{proof}
  As in the proof of Theorem~\ref{thm:selfTest} we can use local unitaries to ensure that $\lambda_0\neq 0$ and $a_0=0$. $P'_0(\gamma)\ket{\Phi_0} = 0$ and $P'_2(\gamma)\ket{\Phi_0} = 0$ then give
\begin{align}
    (2\mu-1)\left( \sin b_1 - \sin a_1  \right) + \sin b_0 = 0, \label{eq:2.1} \\
    (2\mu+1)\left( \sin b_1 + \sin a_1  \right) - \sin b_0 = 0, \label{eq:2.2} \\
    (2\mu-1)\left( \cos b_1 - \cos a_1 \right) + \cos b_0 - 1 = 0, \label{eq:2.3} \\
    (2\mu + 1)\left( \cos b_1 + \cos a_1 \right) - \cos b_0-1 = 0, \label{eq:2.4} 
\end{align}
where $\mu=\cos(\gamma+\pi/6)$.
Eliminating $\sin b_0$ from the first two and $\cos b_0$ from the second two gives
\begin{align}
  \sin a_1&=-2\mu\sin b_1\label{eq:A48}\\
  \cos a_1&=1-2\mu\cos b_1\,.\label{eq:A49}
\end{align}
Using $\sin^2a_1+\cos^2a_1=1$ then gives $\cos b_1=\mu=\cos(\gamma+\pi/6)$ and $\sin b_1=\pm\sin(\gamma+\pi/6)$, corresponding to $B_1=\cos(\gamma+\pi/6)\,\sigma_Z\pm\sin(\gamma+\pi/6)\,\sigma_X$. We can take the case with the minus sign without loss of generality by using the local unitary $\sigma_Z$ if needed.

Eqs.~\eqref{eq:A48} and~\eqref{eq:A49} then give $\sin a_1=\sin(2(\gamma+\pi/6))$ and $\cos a_1=-\cos(2(\gamma+\pi/6))$.

Then~\eqref{eq:2.2} gives
\[\sin b_0=-\left(4\sin^3(\gamma+\pi/6)-3\sin(\gamma+\pi/6)\right)=\sin(3(\gamma+\pi/6))\,,\] and~\eqref{eq:2.4} gives
\[\cos b_0=-\cos(3(\gamma+\pi/6))\,.\]

We hence have
\begin{align*}
  A_0&=\sigma_Z\\
  A_1&=-\cos(2(\gamma+\pi/6))\,\sigma_Z+\sin(2(\gamma+\pi/6))\,\sigma_X\\
  B_0&=-\cos(3(\gamma+\pi/6))\,\sigma_Z+\sin(3(\gamma+\pi/6))\,\sigma_X\\
  B_1&=\cos(\gamma+\pi/6)\,\sigma_Z-\sin(\gamma+\pi/6)\,\sigma_X,
\end{align*}
which is equivalent to the measurement strategy in \cref{eq:gammaStrat_supp}.

The remainder of the argument is identical to that in the proof of \cref{thm:selfTest}.
\end{proof}
\label{sec:gamma}

\subsection{Evaluating the conditional entropy}
By Jordan's lemma, there is no loss in generality if we assume the devices behave according to a convex combination of two-qubit strategies. As proved in the previous sections, the only two-qubit strategy that can saturate \cref{eq:deltaIneq_supp} is that in \cref{eq:deltaStrat_supp} (likewise the only two-qubit strategy that can saturate \cref{eq:gammaIneq_supp} is that in \cref{eq:gammaStrat_supp}), up to local unitaries. Therefore, according to \cref{def:selfTest}, there exists an isometry $V$ from the reference system to the target two-qubit system. For completeness, we now show that the conditional entropy $H(AB|X=0,Y=0,E)$ when the devices maximally saturate one of the self-testing inequalities is equal to entropy of the target strategy unconditioned on Eve. We show this for the $I_{\delta}$-family, and the proof for the $J_{\gamma}$-family is identical.
\begin{theorem}[Entropy of self-tested strategies]
For any physical system achieving $I_{\delta} = I_{\delta}^{\mathrm{Q}}$, its conditional entropy $H(AB|X=0,Y=0,E)_{\rho_{ABE}}$ evaluated for the post-measurement state $\rho_{ABE}$ is given by the entropy of the target strategy unconditioned on $E$, i.e.,
\begin{equation}
    H(AB|X=0,Y=0,E)_{\rho_{ABE}} = H(AB|X=0,Y=0)_{\rho_{AB}} = H(\{\mathrm{p}(ab|00)\}),
\end{equation}
where $\mathrm{p}(ab|00)$ is the distribution of the target two qubit strategy. 
\label{thm:convex} 
\end{theorem}

\begin{proof}
The proof comes directly from the fact that the observation $I_{\delta} = I_{\delta}^{\mathrm{Q}}$ implies the post measurement state is uncorrelated with $E$, and the density operator $\rho_{E}$ can be factored out as a tensor product, i.e. $\rho_{ABE} = \rho_{AB} \otimes \rho_{E}$. The post measurement state for measurements $X=Y=0$ is proportional to
\begin{equation}
    \rho_{ABE} = \sum_{ab}\ketbra{a}{a}_{A} \otimes \ketbra{b}{b}_{B} \\ 
    \otimes \text{Tr}_{\tilde{Q}_{A}\tilde{Q}_{B}}\Big[ (\tilde{M}_{a|0} \otimes \tilde{N}_{b|0} \otimes \mathbb{I}_{E})\ketbra{\Psi}{\Psi}_{\tilde{Q}_{A}\tilde{Q}_{B}E}\Big], \label{eq:state}
\end{equation}
where $\tilde{M}_{a|x},\tilde{N}_{b|y}$ are projectors for the observables $\tilde{A}_{x} = \tilde{M}_{0|x}-\tilde{M}_{1|x}$, $\tilde{B}_{y} = \tilde{N}_{0|y} - \tilde{N}_{1|y}$. From \cref{thm:selfTest}, the observation $I_{\delta} = I_{\delta}^{Q}$ implies the existence of the local isometry $V$ satisfying  
\begin{equation}
    V\Big[ (\tilde{A}_{x} \otimes \tilde{B}_{y} \otimes \mathbb{I}_{E}) \ket{\Psi}_{\tilde{Q}_{A}\tilde{Q}_{B}E} \Big] 
    =   (A_{x} \otimes B_{y}) \ket{\psi}_{Q_{A}Q_{B}} \otimes \ket{\xi}_{\text{Junk}}, 
\end{equation}
in accordance with \cref{def:selfTest}. Since the isometery acts as identity on $E$, we can decompose the junk system as $\mathcal{H}_{\mathrm{Junk}} = \mathcal{H}_{J}\otimes \mathcal{H}_{E}$. Using the fact that $V^{\dagger}V=\mathbb{I}$, we have the following series of equalities for the partial trace term:
\begin{align}
    \text{Tr}_{\tilde{Q}_{A}\tilde{Q}_{B}}\Big[ (\tilde{M}_{a|0} \otimes \tilde{N}_{b|0} \otimes \mathbb{I}_{E})\ketbra{\Psi}{\Psi}_{\tilde{Q}_{A}\tilde{Q}_{B}E}\Big] &= \text{Tr}_{\tilde{Q}_{A}\tilde{Q}_{B}}\Big[ V^{\dagger}V(\tilde{M}_{a|0} \otimes \tilde{N}_{b|0} \otimes \mathbb{I}_{E})V^{\dagger}V\ketbra{\Psi}{\Psi}V^{\dagger}V\Big] \nonumber \\
    &= \text{Tr}_{\tilde{Q}_{A}\tilde{Q}_{B}}\Big[ V^{\dagger}\Big((M_{a|0} \otimes N_{b|0})\ketbra{\psi}{\psi} \otimes \ketbra{\xi}{\xi}\Big)V\Big] \nonumber \\
    &= \text{Tr}_{Q_{A}Q_{B}J}\Big[ (M_{a|0} \otimes N_{b|0})\ketbra{\psi}{\psi} \otimes \ketbra{\xi}{\xi}\Big] \nonumber \\
    &= \mathrm{p}(ab|00)\mathrm{Tr}_{J}\big[\ketbra{\xi}{\xi}\big],
\end{align}
where $\text{p}(ab|xy) = \bra{\psi}M_{a|x} \otimes N_{b|y}\ket{\psi}$ is the distribution generated by the target strategy. Consequently, the post-measurement state takes the form 
\begin{align}
    \rho_{ABE} = \Bigg(\sum_{ab} \text{p}(ab|00)\ketbra{a}{a}\otimes \ketbra{b}{b} \Bigg) \otimes \rho_{E}, \label{eq:postState}
\end{align}
where $\rho_{E} = \mathrm{Tr}_{J}\big[\ketbra{\xi}{\xi}_{JE}\big]$, and we find
\begin{align}
    r = \inf_{\substack{ \rho_{Q_{A}Q_{B}E}, \\ \{M_{a|x}\}_{a}, \{N_{b|y}\}_{b} \\ I_{\delta} = I_{\delta}^{\mathrm{Q}} }} H(AB|X=0,Y=0,E)_{\rho_{ABE}}
    =H(AB|X=0,Y=0)_{\rho_{AB}} 
    = H(\{\text{p}(ab|00)\})\,.
\end{align} 
This concludes the proof. 

\end{proof}

As a corollary of \cref{thm:convex}, $r = 2$ when the $I_{\delta}$ inequalities are used, since $\text{p}(ab|00)=\text{p}_{\delta}(ab|00) = 1/4$ for the self-tested strategies in \cref{eq:deltaStrat_supp}. When the $J_{\gamma}$-family of self-tests are used, $r = 1 + H_{\text{bin}}\Big[\frac{1}{2}(1 + \sin3\gamma )\Big]$, where $H_{\text{bin}}(p) = -p\log_{2}p - (1-p)\log_{2}(1-p)$ is the binary entropy. 
\label{sec:convex}

\section{Proof of Proposition \ifarxiv\ref{prep:maxRand}\else 3\fi}
\label{app:max}
In the main text, we made the following proposition regarding the $I_{\delta}$ and $J_{\gamma}$-family of self-tests described in the previous section:\medskip

\noindent \textbf{Proposition 3} (Maximum randomness versus CHSH value). \textit{The maximum randomness for strategies achieving a CHSH value in the range $s \in (2,3\sqrt{3}/2]$ is $\mathrm{2}$ bits, and is generated by the family of strategies in \cref{eq:deltaStrat_supp}. For the range $s \in [3\sqrt{3}/2,2\sqrt{2}]$, the maximum is given by 
\begin{equation}
    1 + H_{\mathrm{bin}}\Big[\frac{1}{2} + \frac{s}{2}- \frac{3}{\sqrt{2}}\cos\Big(\frac{1}{3} \arccos \Big[-\frac{s}{2\sqrt{2}} \Big]\Big) \Big],
\end{equation}
where $H_{\mathrm{bin}}[\cdot]$ is the binary entropy, and is generated by the family of strategies in \cref{eq:gammaStrat_supp}. }\medskip

This statement is trivial for CHSH scores in the range $(2,3\sqrt{3}/2]$ since each member of the $I_{\delta}$-family generates $r=2$, the global maximum for this scenario. Moreover, the curve provided by the $J_{\gamma}$-family will always be a lower bound on the true maximum, since these are achievable randomness rates certified by the self-tests detailed in \cref{sec:gamma}. In this section we will prove that the $J_{\gamma}$-family give the maximum global randomness achievable by any strategy with the corresponding CHSH value.

Let $\mathcal{Q}$ denote the set of quantum distributions, and $\mathcal{C}(P)$ denote the CHSH value of a distribution $P$. Moreover, let $H(AB|X=0,Y=0,E)_{P}$ be the conditional von Neumann entropy of the outputs $A,B$ for inputs $X=0,Y=0$ given observed distribution $P$, minimized over all quantum strategies that could give rise to $P$, i.e.,
\begin{equation}
   H(AB|X=0,Y=0,E)_{P} := \inf_{\substack{ \rho_{Q_{A}Q_{B}E}, \\ \{M_{a|x}\}_{a}, \{N_{b|y}\}_{b} \\ \text{compatible with} \ P }} H(AB|X=0,Y=0,E)_{\rho_{ABE}}.
\end{equation}
Similarly, let $H(AB|X=0,Y=0)_{P}$ be the Shannon entropy of the distribution on $A,B$ for inputs $X=0,Y=0$.  Then the curve $R:[3\sqrt{3}/2,2\sqrt{2}] \rightarrow [0,2], \ s \mapsto R(s)$ we want to find is defined by the optimization
\begin{align}
    R(s) = \max_P \ &H(AB|X=0,Y=0,E)_{P} \nonumber \\
     \text{s.t.} \  \ & \mathcal{C}(P) = s, \nonumber \\
    & P \in \mathcal{Q}.
\end{align}
Our proof of Proposition \ifarxiv\ref{prep:maxRand} \else 3 \fi proceeds by defining a sequence of upper bounds on $\text{Graph}[R(s)] = \{(s,r) \ | \ r = R(s) \}$, before establishing that the final upper bound is achieved by our $J_{\gamma}$-family of self-tests.

Our first bound follows from strong subadditivity of the von Neumann entropy (that the entropy $H(AB|X=0,Y=0,E)$ cannot decrease if $E$ is discarded) and is $R(s) \leq \bar{R}(s)$, where
\begin{align}
    \bar{R}(s) = \max_P \ &H(AB|X=0,Y=0)_{P} \nonumber \\
     \text{s.t.} \  \ & \mathcal{C}(P) = s, \label{eq:Rsbar} \\
    & P \in \mathcal{Q}.\nonumber
\end{align}

First we prove the following two lemmas:
\begin{lemma}[Monotonicity of $\bar{R}(s)$] The function $\bar{R}(s)$ is strictly decreasing on its domain.  
\label{lem:mon}
\end{lemma}
\begin{proof}
First note that $\bar{R}(3\sqrt{3}/2) > \bar{R}(s) \ \forall s\in (3\sqrt{3}/2,2\sqrt{2}]$. This is because the largest CHSH value achievable when $\text{p}(ab|00) = 1/4$ is $3\sqrt{3}/2$, (see \cref{cor:1}\footnote{Monotonicity of $\bar{R}(s)$ is not needed to establish \cref{cor:1}.}). Because it is an entropy, the objective function $H(AB|X=0,Y=0)_{P}$ is concave in $P$, therefore the optimization~\eqref{eq:Rsbar} defining $\bar{R}(s)$ is convex. It follows that $\bar{R}(s)$ is concave in $s$. To see this, let $\lambda \in [0,1]$ and $s_1,s_2 \in (3\sqrt{3}/2,2\sqrt{2}]$, then
\begin{align}
    \bar{R}[\lambda s_{1} + (1-\lambda)s_{2}] = \max_{P} \ &H(AB|X=0,Y=0)_{P} \nonumber \\
     \text{s.t.} \  \ & \mathcal{C}(P) = \lambda s_{1} + (1-\lambda) s_2, \nonumber \\
    & P \in \mathcal{Q}, \nonumber \\
    \geq \max_{P_{1},P_{2}}\  & H(AB|X=0,Y=0)_{\lambda P_1 + (1-\lambda)P_2} \nonumber \\
    \text{s.t.} \  \ & \mathcal{C}(P_{1}) = s_{1},\ \mathcal{C}(P_{2}) = s_{2}, \nonumber \\
    & P_{1},P_{2} \in \mathcal{Q}, \nonumber \\
    \geq \max_{P_{1},P_{2}}\  & \lambda H(AB|X=0,Y=0)_{P_{1}} + (1-\lambda)H(AB|X=0,Y=0)_{P_{2}} \nonumber \\
    \text{s.t.} \  \ & \mathcal{C}(P_{1}) = s_{1},\ \mathcal{C}(P_{2}) = s_{2}, \nonumber \\
    & P_{1},P_{2} \in \mathcal{Q}, \nonumber \\
    = \lambda \bar{R}( s_{1} & )  +  (1-\lambda)\bar{R}(s_{2}),
\end{align}
where we used the concavity of the Shannon entropy to obtain the inequality. Since $\bar{R}(s)$ is initially decreasing, and is a concave function, it must be monotonically decreasing.   
\end{proof}

\begin{lemma}[Inverse function of $\bar{R}(s)$]\label{lem:inv}
    Suppose $r = \bar{R}(s)$. The function $\bar{R}(s)$ has the following inverse, denoted $\bar{R}^{-1}$, that satisfies $s = \bar{R}^{-1}(r)$, given by
    \begin{align}
    \bar{R}^{-1}(r) = \max \ &\mathcal{C}(P) \nonumber \\
     \mathrm{s.t.} \  \ & H(AB|X=0,Y=0)_{P} = r, \nonumber \\
    & P \in \mathcal{Q}. \label{eq:invR}
\end{align}
\end{lemma}
\begin{proof}
We prove \cref{lem:inv} by showing $\bar{R}^{-1}(\bar{R}(s)) = s$, and $\bar{R}(\bar{R}^{-1}(r)) = r$, and using \cref{lem:mon}. First consider $\bar{R}^{-1}(\bar{R}(s)) = s$,
\begin{align}
    \bar{R}^{-1}(\bar{R}(s)) = \max \ &\mathcal{C}(P) \nonumber \\
     \text{s.t.} \  \ & H(AB|X=0,Y=0)_{P} = \bar{R}(s), \nonumber \\
    & P \in \mathcal{Q}.
\end{align}
The constraint $H(AB|X=0,Y=0)_{P} = \bar{R}(s)$ implies the achievable CHSH values for the distribution $P$ must lie to the left of $s$, i.e., $\mathcal{C}(P) \leq s$, since the curve $\bar{R}(s)$ is decreasing (cf.\ Lemma~\eqref{lem:mon}). We therefore have that $\bar{R}^{-1}(\bar{R}(s)) = \max_{\{P\in \mathcal{Q} \ \text{s.t.} \ \mathcal{C}(P) \leq s\}}\mathcal{C}(P) = s$. For the other direction $\bar{R}(\bar{R}^{-1}(r))$, the same reasoning holds. The constraint $\mathcal{C}(P) = \bar{R}^{-1}(r)$ implies that $ H(AB|X=0,Y=0)_{P} \leq r$ since any distribution that achieves a CHSH value of $\bar{R}^{-1}(r)$ can generate no more than $r$ bits of randomness. Hence $\bar{R}(\bar{R}^{-1}(r)) = r$. This completes the proof.
\end{proof}
From the above lemma, we can solve for upper bounds on the points $(s,R(s)) \in \text{Graph}[R(s)]$ using the inverse function, i.e., $(s,\bar{R}(s)) = (\bar{R}^{-1}(r),r)$ where $\bar{R}(s) = r$. What remains is to compute $\bar{R}^{-1}(r)$ (or at least an upper bound, which will correspond to an upper bound on $R(s)$ due to the monotonicity argument). To do so we use the following two lemmas to formulate the constraints $H(AB|X=0,Y=0)_{P} = r$ as linear functions of the distribution $P$, defining a new upper bound:
\begin{lemma}
    Let $\mathcal{E}$ be the local channel that flips both output bits with probability $1/2$, i.e., $\mathcal{E}:\{\text{p}(ab|xy)\} \rightarrow \Big \{\frac{1}{2}\text{p}(ab|xy)+\frac{1}{2}\text{p}(\bar{a}\bar{b}|xy)\Big\}$ where $\bar{a}$ ($\bar{b}$) is the bit-wise complement of $a$ ($b$), i.e., $\bar{a}=a\oplus1$. The entropy after applying $\mathcal{E}$ is non-decreasing. Further, the CHSH value is invariant under $\mathcal{E}$.   
    \label{lem:entSimp}
\end{lemma}
\begin{proof}
The first claim comes from the data processing inequality, that states that the entropy is non-decreasing under post-processing, i.e., $H(AB|X=0,Y=0)_{P} \leq H(AB|X=0,Y=0)_{\mathcal{E}(P)}$. The second claim comes from the fact that the correlators $\langle A_{x} B_{y} \rangle$ are invariant under $\mathcal{E}$. 
\end{proof}
Notice that when Alice applies the post-processing map $\mathcal{E}$ to her devices, the probabilities are symmetrized, i.e., $\text{p}(aa|00) = \epsilon$, $\text{p}(a\bar{a}|00) = 1/2 - \epsilon$, $0 \leq \epsilon \leq 1/2 $. In this case, we find $H(AB|X=0,Y=0)_{\mathcal{E}(P)} = 1+H_{\text{bin}}(2\epsilon)$. As a consequence of \cref{lem:entSimp}, we can define the following upper bound on $\bar{R}(s)$:
\begin{align}
    \bar{R}(s) \leq \bar{\bar{R}}(s) =  \max \ &H(AB|X=0,Y=0)_{\mathcal{E}(P)} \nonumber \\
     \text{s.t.} \  \ & \mathcal{C}(\mathcal{E}(P)) = s, \nonumber \\
    & P \in \mathcal{Q} \nonumber \\
    = \max \ &H(AB|X=0,Y=0)_{P} \nonumber \\
     \text{s.t.} \  \ & \mathcal{C}(P) = s, \nonumber \\
     & \text{p}(00|00) = \text{p}(11|00)  , \nonumber \\
     & \text{p}(01|00) = \text{p}(10|00) , \nonumber \\
    & P \in \mathcal{Q},
\end{align}
where the second equality comes from the fact that optimizing the entropy over $\mathcal{E}(P),P\in \mathcal{Q}$ is equal to optimizing the entropy over symmetrized distributions in $\mathcal{Q}$ (following the convexity of $\mathcal{Q}$), and $\mathcal{C}(\mathcal{E}(P)) = \mathcal{C}(P)$. We can then define an inverse using \cref{lem:inv}, just as was done for $\bar{R}(s)$; we remark that \cref{lem:inv} applies here, since \cref{lem:mon,lem:inv} will hold when $\mathcal{Q}$ is replaced by any convex subset of $\mathcal{Q}$, e.g., the set of symmetrized quantum distributions. This inverse is given by 
\begin{align}
    \bar{\bar{R}}^{-1}(r) = \max \ &\mathcal{C}(P) \nonumber \\
     \text{s.t.} \  \ & H(AB|X=0,Y=0)_{P} = r, \nonumber \\
     & \text{p}(00|00) = \text{p}(11|00)  , \nonumber \\
     & \text{p}(01|00) = \text{p}(10|00) , \nonumber \\
    & P \in \mathcal{Q}. \label{eq:invSym}
\end{align}

\begin{lemma}
    The optimization in \cref{eq:invSym} has the following upper bound:
    \begin{align}
     \bar{\bar{R}}^{-1}(r) \leq \max \ &\mathcal{C}(P) \nonumber \\
     \mathrm{s.t.} \  \ & \langle A_{0}B_{0} \rangle = 4\epsilon_{r} - 1 , \nonumber \\
    & P \in \mathcal{Q},
    \label{eq:RUB1}
    \end{align}
    where $\epsilon_{r}$ satisfies $r = 1 + H_{\mathrm{bin}}(2\epsilon_{r})$.
    \label{lem:corrUB}
\end{lemma}
\begin{proof}
Firstly, consider symmetrized distributions, i.e., $\text{p}(aa|00) = \epsilon, \ \text{p}(a\bar{a}|00) = 1/2 - \epsilon$, $H(AB|X=0,Y=0)_{P} = 1+H_{\text{bin}}(2\epsilon)$. One can notice that for $\epsilon \in [1/4,1/2]$ there is a one to one mapping between $H(AB|X=0,Y=0)_{P}$ and $\epsilon$. Moreover, the range of $\epsilon$ we are interested in is given by $\epsilon \in [1/4,(2+\sqrt{2})/8]$, since $\epsilon = 1/4$ corresponds to the $\delta = \pi/6$ strategy ($r=2$), and $\epsilon=(2+\sqrt{2})/8$ corresponds to the optimal CHSH strategy ($r\approx 1.6$). Hence for every choice of $r$, there exists a unique $\epsilon_{r}$ that satisfies $r = H(AB|X=0,Y=0)_{P} = 1 + H_{\text{bin}}(2\epsilon_{r})$ for $r\in [2,1.6...]$. We can therefore write the constraint $r = H(AB|X=0,Y=0)_{P}$ in terms of linear functions of $P$:  
\begin{align}
     \bar{\bar{R}}^{-1}(r) = \max \ &\mathcal{C}(P) \nonumber \\
     \text{s.t.} \  \ & \text{p}(00|00) = \epsilon_{r} , \nonumber \\
     & \text{p}(01|00) = \text{p}(10|00) =\frac{1}{2} - \epsilon_{r} , \nonumber \\
    & P \in \mathcal{Q},
\end{align}
where $\epsilon_{r}$ satisfies $r = 1 + H_{\text{bin}}(2\epsilon_{r})$. We can now relax this by considering the two party correlators, $\langle A_{x}B_{y}\rangle$; we replace the stronger constraints on the probabilities $\text{p}(ab|00)$ with a single weaker constraint on the $X=0,Y=0$ correlator, and arrive at the desired upper bound.
\end{proof}

In the next two lemmas, we rewrite the upper bound in \cref{eq:RUB1} using an SOS decomposition. Let $C = A_{0}B_{0} + A_{0}B_{1} + A_{1}B_{0} - A_{1}B_{1}$ be the CHSH operator, and consider the following optimization:
\begin{align}
    \tilde{R}^{-1}(r) = \min_{t,z,S} \ &t \nonumber \\
     \text{s.t.} \  \ & t\mathbb{I} - C = S + z\Big(A_{0}B_{0} - (4\epsilon_{r} - 1)\mathbb{I} \Big) , \label{eq:RUB2}
\end{align}
where $S$ is an SOS decomposition for the operator expression $t\mathbb{I} - z\Big(A_{0}B_{0} - (4\epsilon_{r} - 1)\mathbb{I} \Big) - C$. One can notice that for any feasible point $(t,z,S)$, and any distribution $P$ that satisfies $\langle A_{0}B_{0} \rangle = 4\epsilon_{r} - 1 $, we get an upper bound on the CHSH value, $t \geq \langle C \rangle = \mathcal{C}(P)$. Hence $\tilde{R}^{-1}(r)$ gives an upper bound on the CHSH value across all distributions that satisfy $\langle A_{0}B_{0} \rangle = 4\epsilon_{r} - 1 $, i.e., an upper bound on \cref{eq:RUB1}. An SOS decomposition is given in the following lemma:
\begin{lemma}
    Let $\bm{R}$ be as defined in~\eqref{lem:SOS_delta}. The operator expression $t\mathbb{I} - z\Big(A_{0}B_{0} - (4\epsilon_{r} - 1)\mathbb{I} \Big) - C$ admits the SOS decomposition
    \begin{equation}
    M = 
    \begin{bmatrix}
        m_{2} - 1/2  &  1/2 &  0 & 0 \\
        1/2 & m_{1} + (z+1)/2 & 0 & 0 \\
        0 & 0 &  m_{2} + 1/2 & -1/2 \\
        0 & 0 & -1/2 & m_{1} - (z+1)/2
    \end{bmatrix},
\end{equation}
for any $m_{1},m_{2}$ satisfying $2(m_{1} + m_{2}) = t + z(4\epsilon_{r} -1)$. \label{lem:RSOS}
\end{lemma}
This was derived using the symmetry arguments as was done for self-testing, and one can verify for any $m_{1},m_{2}$ that satisfy the equality condition $\bm{R}^{\dagger}M\bm{R} = t\mathbb{I} - z\Big(A_{0}B_{0} - (4\epsilon_{r} - 1)\mathbb{I} \Big) - C$. 
\begin{lemma}
    The upper bound in \cref{eq:RUB2} is equivalent to the the following optimization problem:
    \begin{align}
        \tilde{R}^{-1}(r) = \max_{\mu} \ \ & \sqrt{(2-4\epsilon_{r})^{2}+ (2-4\epsilon_{r})\mu} + \sqrt{(4\epsilon_{r})^{2} - 4\epsilon_{r}\mu} + \mu \nonumber \\
    &4\epsilon_{r} - 2 \leq \mu \leq 4\epsilon_{r}.
\end{align}
Moreover, the optimal value is given by
\begin{equation}
    \tilde{R}^{-1}(r) = 6\cos(2\theta)-4\cos^3(2\theta), 
\end{equation}
for the optimal argument
\begin{equation}
    \mu^*=4\cos(2\theta)\sin^2(2\theta),
\end{equation}
where $\theta = \frac{1}{6} \arccos(1-4\epsilon_{r})$.
\label{lem:finalOpt}
\end{lemma}
\begin{proof}
By inserting the SOS decomposition from \cref{lem:RSOS}, we can rewrite the optimization in the following way: 
\begin{align}
    \tilde{R}^{-1}(r) = \min_{t,z,m_{1},m_{2}} \ &t \nonumber \\
     \text{s.t.} \  \ & 
     \begin{bmatrix}
        m_{2} - 1/2  &  1/2  \\
        1/2 & m_{1} + (z+1)/2 
    \end{bmatrix} 
    \succeq 0, \nonumber \\
    &\begin{bmatrix}
        m_{2} + 1/2  &  -1/2  \\
        -1/2 & m_{1} - (z+1)/2 
    \end{bmatrix} 
    \succeq 0, \nonumber \\
    &2(m_{1} + m_{2}) = t + z(4\epsilon_{r} -1) \nonumber \\
    = \min_{X_{1},X_{2}} \ & (2-4\epsilon_{r})\text{Tr}[X_{1}] + 4\epsilon_{r}\text{Tr}[X_{2}] \nonumber \\
    \text{s.t.} \ \ &
    \text{Tr}[X_{1} \ketbra{0}{1}] = 1/2, \nonumber \\
    &\text{Tr}[X_{2} \ketbra{0}{1}] = -1/2, \nonumber \\
    &\text{Tr}[X_{1}\ketbra{0}{0}] - \text{Tr}[X_{2}\ketbra{0}{0}] = -1, \nonumber \\
    & X_{1} \succeq 0, \ X_{2} \succeq 0,
    \label{eq:primal}
\end{align}
where $\{\ket{i}\}$ is the standard computational basis. We remark that the particular form of the SOS decomposition used is not unique, and strictly speaking we therefore find and upper bound on $\tilde{R}^{-1}(s)$ when inserting this into the constraint. For ease of notation we redefine $\tilde{R}^{-1}(s)$ above and acknowledge this is an upper bound on \cref{eq:RUB2}. This SDP has the following dual:
\begin{align}
    \max_{\lambda,\nu,\mu} \ &\lambda + \nu + \mu \nonumber \\
     \text{s.t.} \  \ & 
     \begin{bmatrix}
       2 - 4\epsilon_{r}  &  -\lambda  \\
        -\lambda & 2 - 4\epsilon_{r} + \mu 
    \end{bmatrix} 
    \succeq 0, \nonumber \\
    &\begin{bmatrix}
        4\epsilon_{r}  &  -\nu  \\
        -\nu & 4\epsilon - \mu 
    \end{bmatrix} 
    \succeq 0, \nonumber \\
    = \max_{\lambda,\nu,\mu} \ \ &\lambda + \nu + \mu \nonumber \\
    &4\epsilon_{r} - 2 \leq \mu \leq 4\epsilon_{r} \nonumber , \\
    &\lambda^{2} \leq (4\epsilon_{r}-2)^{2}- (4\epsilon_{r}-2)\mu \nonumber , \\
    & \nu^{2} \leq 4\epsilon_{r}(4\epsilon_{r} - \mu) \nonumber , \\
    = \max_{\mu} \ \ & \sqrt{(4\epsilon_{r}-2)^{2}- (4\epsilon_{r}-2)\mu} + \sqrt{(4\epsilon_{r})^{2} - 4\epsilon_{r}\mu} + \mu \nonumber \\
    &4\epsilon_{r} - 2 \leq \mu \leq 4\epsilon_{r},\label{eq:final_opt}
\end{align}
where the last equality comes from the fact that the objective is maximized when $\lambda$ and $\nu$ saturate their respective upper bounds. The first claim then follows from strong duality. To see this, consider the primal problem in \cref{eq:primal}; the point $(t,z,m_{1},m_{2}) = (4,0,1,1)$ satisfies $2(m_{1}+m_{2}) = t + z(4\epsilon_{r} - 1)$, and the eigenvalues of the two matrices are given by $1 \pm 1/\sqrt{2} > 0$. This point is strictly feasible, i.e., Slater's condition holds. 

For the second claim, consider the final optimization~\eqref{eq:final_opt} over $\mu$. For algebraic convenience, let us use the shifted variable $m=\mu-4\epsilon_r+1$, with $-1\leq m\leq1$. Let $f(m)$ be the objective function in terms of $m$, i.e., $f(m)=m+4\epsilon_r-1+\sqrt{2(1+m)(1-2\epsilon_r)}+2\sqrt{\epsilon_r(1-m)}$. Then $f'(m)=0$ gives
\[
1-\frac{\epsilon_r}{\sqrt{\epsilon_r(1-m)}}=\frac{1-2\epsilon_r}{\sqrt{2(1+m)(1-2\epsilon_r}}.\]
After some rearrangement we obtain
\[\frac{(m+4\epsilon_r-1)(4m^3-3m+4\epsilon_r-1)}{1-m^2}=0.\]
The solutions are hence $m=1-4\epsilon_r$ or $m$ needs to be a solution of the cubic $4m^3-3m+4\epsilon_r-1=0$. Using the formula for the roots of a cubic we find the roots to be
\[
m_k=\cos\left(\frac{1}{3}\arccos(1-4\epsilon_r)-\frac{2\pi k}{3}\right)\quad\text{where}\quad k=0,1,2.
\]
Considering the four stationary points and the two endpoints of the range of $m$ we find that the maximum occurs for $m=m^*=\cos\left(\frac{1}{3}\arccos(1-4\epsilon_r)\right)$, which corresponds to $\mu=\mu^*=m^*+4\epsilon_r-1$. If we define $\theta=\frac{1}{6}\arccos(1-4\epsilon_r)$ so that $m^*=\cos(2\theta)$, $1-4\epsilon_r=\cos(6\theta)$ and $2\epsilon_r=\sin^2(3\theta)$, we find $\mu^*=\cos(2\theta)-\cos(6\theta)$. The maximum value of the objective function is
\begin{align*}
    f(m^*)&=\cos(2\theta)-\cos(6\theta)+\sqrt{2\cos^2(3\theta)(1+\cos(2\theta)))}+\sqrt{2\sin^2(3\theta)(1-\cos(2\theta))}\\
    &=\cos(2\theta)-\cos(6\theta)+2\cos(\theta)\cos(3\theta)+2\sin(\theta)\sin(3\theta)\\
    &=3\cos(2\theta)-\cos(6\theta)=6\cos(2\theta)-4\cos^3(2\theta)\,,
    \end{align*}
where we have used that for $\epsilon_r\in[1/4,1/2]$, $\theta\in[\pi/12,\pi/6]$ so that $\cos(\theta)$, $\sin(\theta)$, $\cos(3\theta)$ and $\sin(3\theta)$ are all positive.
\end{proof}
\begin{corollary}\label{cor:1}
The maximum CHSH score achievable by any quantum strategy with $\text{p}(ab|00) = 1/4$ for all $a$ and $b$ is $3\sqrt{3}/2$.
\end{corollary}
\begin{proof}
  When $\epsilon_{r} = 1/4$, $\tilde{R}^{-1}(r) = 3\sqrt{3}/2$, i.e., $3\sqrt{3}/2$ is an upper bound on the maximum achievable CHSH value when $\text{p}(ab|00) = 1/4$. We know this upper bound is achievable for the $\delta = \pi/6$ self-test in \cref{eq:deltaIneq_supp}, hence this must be the true maximum.
\end{proof}

Our final theorem shows the optimality of the constructions.
\begin{theorem}[Maximal global randomness versus CHSH value]
    The maximum global randomness, $R(s)$, for quantum strategies that achieve a particular CHSH value $s$ is given by
    \begin{equation}
        R(s) = 
        \begin{cases}
			2, & s \in (2,3\sqrt{3}/2 ] \\
            1 + H_{\mathrm{bin}}\Big[\frac{1}{2} + \frac{s}{2}- \frac{3}{\sqrt{2}}\cos\Big(\frac{1}{3} \arccos \Big[-\frac{s}{2\sqrt{2}} \Big]\Big) \Big], & s \in [3\sqrt{3}/2,2\sqrt{2}],
		 \end{cases}
    \end{equation}
    where $H_{\mathrm{bin}}(p) = -p\log_{2}p - (1-p)\log_{2}(1-p)$ is the binary entropy. Moreover, the inequalities in \cref{eq:deltaIneq_supp} and \cref{eq:gammaIneq_supp} self-test the quantum state and measurements that achieve this maximum. 
    \label{thm:maxRand}
\end{theorem}
\begin{proof}
The case of $R(s)=2$ for $s \in (2,3\sqrt{3}/2]$ is trivially an upper bound on the maximum, and is shown to be achievable by the self-tests in \cref{eq:deltaIneq_supp}. For the case $s \in [3\sqrt{3}/2,2\sqrt{2}]$, we use the sequence of upper bounds and inverse functions defined in this section. Consider the points $(s,r) \in \text{Graph}[R(s)]$. We have the following:
\begin{align}
    (s,r) &= (s,R(s)) \nonumber \\
    &\leq(s,\bar{R}(s)) \nonumber \\
    &\leq(s,\bar{\bar{R}}(s)) \nonumber \\
    &= (\bar{\bar{R}}^{-1}(r),r) \nonumber \\
    & \leq (\tilde{R}^{-1}(r),r),
\end{align} 
where $\leq$ denotes component-wise inequality. Hence we can find an upper bound on $\text{Graph}[R(s)]$ using \cref{lem:finalOpt}, i.e.
\begin{equation}
    s \leq \tilde{R}^{-1}(r) = 6\cos 2\theta-4\cos^3 2\theta, \ \theta = \frac{1}{6} \arccos[1-4\epsilon_{r}], \ r = 1+H_{\text{bin}}(2\epsilon_{r}).
\end{equation}
We define this upper bound on $\text{Graph}[R(s)]$ as
\begin{equation}
    \text{Graph}[\tilde{R}^{-1}(r)] = \Big\{ (s,r) \ | \ s = \tilde{R}^{-1}(r) \Big\}.
\end{equation}
We now show it is achievable. From the self-tests in \cref{eq:gammaIneq_supp}, we find a tight lower bound on the conditional von Neumann entropy parameterized by $\gamma \in [0,\pi/12]$,
\begin{align}
   r(\gamma) &= \inf_{\substack{ \rho_{Q_{A}Q_{B}E}, \\ \{M_{a|x}\}_{a}, \{N_{b|y}\}_{b} \\ \text{Compatible with:} \ \langle S_{\gamma} \rangle = I_{\gamma}^{\text{Q}} }} H(AB|X=0,Y=0,E)_{\rho_{ABE}} \nonumber \\
   &= H(AB|X=0,Y=0)_{P_{\gamma}} \nonumber \\
   &= 1 + H_{\text{bin}}\Big[\frac{1}{2}(1+\sin 3\gamma)\Big],
\end{align}
where $P_{\gamma}$ is the distribution generated by \cref{eq:gammaStrat_supp}. We find the associated CHSH value is given by 
\begin{equation}
    s(\gamma) = \mathcal{C}(P_{\gamma}) = \sin 3\gamma + 3\cos \Big( \gamma + \frac{\pi}{6}\Big).
    \label{eq:gammaS}
\end{equation}
Since this is achievable, we have derived a parametric lower bound on $\text{Graph}[R(s)] = \{(s,r) \ | \ r = R(s)\}$: 
\begin{equation}
    \text{Graph}_{\Gamma} =  \Big\{(s(\gamma),r(\gamma)) \ | \ \gamma \in [0,\pi/12]\Big\}.
\end{equation}
Analysing the $X=0,Y=0$ block of $P_{\gamma}$, we find $\epsilon_{r} = \frac{1}{4}(1+\sin 3\gamma)$. Inverting this, we find $\gamma = \frac{1}{3}\arcsin[4\epsilon_{r} - 1]$, and inserting into \cref{eq:gammaS}, we express $s$ in terms of $\epsilon_{r}$, and hence $r$. Calling this function $R_{\Gamma}^{-1}(r)$:
\begin{align}
    R_{\Gamma}^{-1}(r) &\equiv s(\gamma) = 4\epsilon_{r} - 1 + 3\cos\Big(\frac{1}{3}\arcsin[4\epsilon_{r} - 1] + \frac{\pi}{6}\Big )  \nonumber \\
    &= -\cos 6\theta  +3\cos 2\theta = 6\cos 2\theta-4\cos^3 2\theta,
\end{align}
where we used the identities $\arcsin(x) = -\arcsin(-x)$ and $\arcsin(x) = \pi/2 - \arccos(x)$. This implies
\begin{equation}
    \text{Graph}_{\Gamma} = \text{Graph}[R_{\Gamma}^{-1}(r)] =  \Big\{(s,r) \ | \ s = R_{\Gamma}^{-1}(r)\Big\}.
\end{equation}
One can immediately see that $R_{\Gamma}^{-1}(r) = \tilde{R}^{-1}(r)$, from which it follows $\text{Graph}[\tilde{R}^{-1}(r)] = \text{Graph}[R_{\Gamma}^{-1}(r)]$, i.e., the upper and lower bounds coincide, and $\text{Graph}[R(s)] = \text{Graph}[\tilde{R}^{-1}(r)]=\text{Graph}[R_{\Gamma}^{-1}(r)]$. This shows that the family of inequalities in \cref{eq:gammaIneq_supp} self-test the maximum. 

From this, we can derive an explicit expression for $R(s)$, $s\in[3\sqrt{3}/2,2\sqrt{2}]$. We begin by changing variables $\hat{\theta} = 2\theta = \frac{1}{3}\arccos[1-4\epsilon_{r}]$. We wish to express $s$ in terms of $\hat{\theta}$, and hence $\epsilon_{r}$, which amounts to solving the cubic
\begin{equation}
    4\cos^{3}\hat{\theta} - 6\cos \hat{\theta} + s = 0.
\end{equation}
Employing another change of variables, $\cos \hat{\theta} = \sqrt{2} \cos \phi$:
\begin{align}
    4\cos^{3} \phi -3\cos \phi = \cos 3\phi = - \frac{s}{2\sqrt{2}},
\end{align}
which has solutions
\begin{equation}
    \phi_{k} = \frac{1}{3} \arccos \Big[ -\frac{s}{2\sqrt{2}} \Big] + \frac{2\pi k}{3}, \ k=0,1,2.
\end{equation}
Notice for $\theta \in [\pi/12,\pi/8]$, we require $\sqrt{2}\cos \phi \in [\sqrt{3}/2,1/\sqrt{2}]$, which for $s\in [3\sqrt{3}/2,2\sqrt{2}]$ is only satisfied when $k=0$. We therefore find that $\cos \hat{\theta} = \sqrt{2}\cos\Big(\frac{1}{3} \arccos \Big[ -\frac{s}{2\sqrt{2}} \Big]\Big)$. We can now solve for $\epsilon_{r}$, setting $\phi \equiv \phi_{0}$:
\begin{align}
    \epsilon_{r} &= \frac{1}{4} (1 - \cos 3\hat{\theta}) \nonumber \\
    &= \frac{1}{4}(1 - 4\cos^{3}\hat{\theta} + 3\cos \hat{\theta} ) \nonumber \\
    &= \frac{1}{4}\Big( 1 - \sqrt{2}( 1- 4 \cos^{3} \phi + 3 \cos \phi - 4\cos^{3} \phi) \Big) \nonumber \\
    &= \frac{1}{4}\Big( 1 - \sqrt{2} + 2\sqrt{2} \sin^{2} \Big(\frac{3\phi}{2} \Big) - 4 \sqrt{2} \cos^{3}\phi \Big) \nonumber \\
    &= \frac{1}{4}\Big( 1 + s- 3\sqrt{2}\cos\Big[\frac{1}{3} \arccos \Big( -\frac{s}{2\sqrt{2}} \Big)\Big] \Big),
\end{align}
where for the second equality we used the identity $\cos 3\theta = 4 \cos^{3}\theta - 3\cos \theta$, for third we used $1-4\cos^{3}\theta + 3\cos \theta = 2 \sin^{2} \Big( \frac{3\theta}{2}\Big)$, $\sin^{2}\theta = \frac{1}{2}(1-\cos2\theta)$, and for the final we used the triple angle formula again. The claim then follows using the fact that $R(s) = 1 + H_{\text{bin}}(2\epsilon_{r})$.  
\end{proof}


\begin{thebibliography}{54}%
\makeatletter
\providecommand \@ifxundefined [1]{%
 \@ifx{#1\undefined}
}%
\providecommand \@ifnum [1]{%
 \ifnum #1\expandafter \@firstoftwo
 \else \expandafter \@secondoftwo
 \fi
}%
\providecommand \@ifx [1]{%
 \ifx #1\expandafter \@firstoftwo
 \else \expandafter \@secondoftwo
 \fi
}%
\providecommand \natexlab [1]{#1}%
\providecommand \enquote  [1]{``#1''}%
\providecommand \bibnamefont  [1]{#1}%
\providecommand \bibfnamefont [1]{#1}%
\providecommand \citenamefont [1]{#1}%
\providecommand \href@noop [0]{\@secondoftwo}%
\providecommand \href [0]{\begingroup \@sanitize@url \@href}%
\providecommand \@href[1]{\@@startlink{#1}\@@href}%
\providecommand \@@href[1]{\endgroup#1\@@endlink}%
\providecommand \@sanitize@url [0]{\catcode `\\12\catcode `\$12\catcode
  `\&12\catcode `\#12\catcode `\^12\catcode `\_12\catcode `\%12\relax}%
\providecommand \@@startlink[1]{}%
\providecommand \@@endlink[0]{}%
\providecommand \url  [0]{\begingroup\@sanitize@url \@url }%
\providecommand \@url [1]{\endgroup\@href {#1}{\urlprefix }}%
\providecommand \urlprefix  [0]{URL }%
\providecommand \Eprint [0]{\href }%
\providecommand \doibase [0]{http://dx.doi.org/}%
\providecommand \selectlanguage [0]{\@gobble}%
\providecommand \bibinfo  [0]{\@secondoftwo}%
\providecommand \bibfield  [0]{\@secondoftwo}%
\providecommand \translation [1]{[#1]}%
\providecommand \BibitemOpen [0]{}%
\providecommand \bibitemStop [0]{}%
\providecommand \bibitemNoStop [0]{.\EOS\space}%
\providecommand \EOS [0]{\spacefactor3000\relax}%
\providecommand \BibitemShut  [1]{\csname bibitem#1\endcsname}%
\let\auto@bib@innerbib\@empty
\bibitem [{\citenamefont {Bell}(1987)}]{Bell_book}%
  \BibitemOpen
  \bibfield  {author} {\bibinfo {author} {\bibfnamefont {J.~S.}\ \bibnamefont
  {Bell}},\ }\href {\doibase 10.1017/CBO9780511815676} {\emph {\bibinfo {title}
  {Speakable and unspeakable in quantum mechanics}}}\ (\bibinfo  {publisher}
  {Cambridge University Press},\ \bibinfo {year} {1987})\BibitemShut {NoStop}%
\bibitem [{\citenamefont {Brunner}\ \emph {et~al.}(2014)\citenamefont
  {Brunner}, \citenamefont {Cavalcanti}, \citenamefont {Pironio}, \citenamefont
  {Scarani},\ and\ \citenamefont {Wehner}}]{Brunner_review}%
  \BibitemOpen
  \bibfield  {author} {\bibinfo {author} {\bibfnamefont {N.}~\bibnamefont
  {Brunner}}, \bibinfo {author} {\bibfnamefont {D.}~\bibnamefont {Cavalcanti}},
  \bibinfo {author} {\bibfnamefont {S.}~\bibnamefont {Pironio}}, \bibinfo
  {author} {\bibfnamefont {V.}~\bibnamefont {Scarani}}, \ and\ \bibinfo
  {author} {\bibfnamefont {S.}~\bibnamefont {Wehner}},\ }\bibfield  {title}
  {\enquote {\bibinfo {title} {Bell nonlocality},}\ }\href {\doibase
  10.1103/revmodphys.86.419} {\bibfield  {journal} {\bibinfo  {journal}
  {Reviews of Modern Physics}\ }\textbf {\bibinfo {volume} {86}},\ \bibinfo
  {pages} {419–478} (\bibinfo {year} {2014})}\BibitemShut {NoStop}%
\bibitem [{\citenamefont {Mayers}\ and\ \citenamefont
  {Yao}(2004)}]{mayers2004self}%
  \BibitemOpen
  \bibfield  {author} {\bibinfo {author} {\bibfnamefont {D.}~\bibnamefont
  {Mayers}}\ and\ \bibinfo {author} {\bibfnamefont {A.}~\bibnamefont {Yao}},\
  }\href@noop {} {\enquote {\bibinfo {title} {Self testing quantum
  apparatus},}\ } (\bibinfo {year} {2004}),\ \Eprint
  {http://arxiv.org/abs/quant-ph/0307205} {arXiv:quant-ph/0307205 [quant-ph]}
  \BibitemShut {NoStop}%
\bibitem [{\citenamefont {McKague}\ \emph {et~al.}(2012)\citenamefont
  {McKague}, \citenamefont {Yang},\ and\ \citenamefont
  {Scarani}}]{McKagueSinglet}%
  \BibitemOpen
  \bibfield  {author} {\bibinfo {author} {\bibfnamefont {M.}~\bibnamefont
  {McKague}}, \bibinfo {author} {\bibfnamefont {T.~H.}\ \bibnamefont {Yang}}, \
  and\ \bibinfo {author} {\bibfnamefont {V.}~\bibnamefont {Scarani}},\
  }\bibfield  {title} {\enquote {\bibinfo {title} {Robust self-testing of the
  singlet},}\ }\href {\doibase 10.1088/1751-8113/45/45/455304} {\bibfield
  {journal} {\bibinfo  {journal} {Journal of Physics A: Mathematical and
  Theoretical}\ }\textbf {\bibinfo {volume} {45}},\ \bibinfo {pages} {455304}
  (\bibinfo {year} {2012})}\BibitemShut {NoStop}%
\bibitem [{\citenamefont {Yang}\ and\ \citenamefont
  {Navascués}(2013)}]{YangSelfTest}%
  \BibitemOpen
  \bibfield  {author} {\bibinfo {author} {\bibfnamefont {T.~H.}\ \bibnamefont
  {Yang}}\ and\ \bibinfo {author} {\bibfnamefont {M.}~\bibnamefont
  {Navascués}},\ }\bibfield  {title} {\enquote {\bibinfo {title} {Robust
  self-testing of unknown quantum systems into any entangled two-qubit
  states},}\ }\href {\doibase 10.1103/PhysRevA.87.050102} {\bibfield  {journal}
  {\bibinfo  {journal} {Physical Review A}\ }\textbf {\bibinfo {volume} {87}},\
  \bibinfo {pages} {050102} (\bibinfo {year} {2013})}\BibitemShut {NoStop}%
\bibitem [{\citenamefont {Kaniewski}(2017)}]{KaniewskiSelfTest}%
  \BibitemOpen
  \bibfield  {author} {\bibinfo {author} {\bibfnamefont {J.}~\bibnamefont
  {Kaniewski}},\ }\bibfield  {title} {\enquote {\bibinfo {title} {Self-testing
  of binary observables based on commutation},}\ }\href {\doibase
  10.1103/PhysRevA.95.062323} {\bibfield  {journal} {\bibinfo  {journal}
  {Physical Review A}\ }\textbf {\bibinfo {volume} {95}},\ \bibinfo {pages}
  {062323} (\bibinfo {year} {2017})}\BibitemShut {NoStop}%
\bibitem [{\citenamefont {{\v S}upi\'c}\ and\ \citenamefont
  {Bowles}(2020)}]{SupicSelfTest}%
  \BibitemOpen
  \bibfield  {author} {\bibinfo {author} {\bibfnamefont {I.}~\bibnamefont {{\v
  S}upi\'c}}\ and\ \bibinfo {author} {\bibfnamefont {J.}~\bibnamefont
  {Bowles}},\ }\bibfield  {title} {\enquote {\bibinfo {title} {Self-testing of
  quantum systems: a review},}\ }\href {\doibase 10.22331/q-2020-09-30-337}
  {\bibfield  {journal} {\bibinfo  {journal} {Quantum}\ }\textbf {\bibinfo
  {volume} {4}},\ \bibinfo {pages} {337} (\bibinfo {year} {2020})}\BibitemShut
  {NoStop}%
\bibitem [{\citenamefont {Barrett}\ \emph
  {et~al.}(2005{\natexlab{a}})\citenamefont {Barrett}, \citenamefont {Linden},
  \citenamefont {Massar}, \citenamefont {Pironio}, \citenamefont {Popescu},\
  and\ \citenamefont {Roberts}}]{BarrettNonlocalResource}%
  \BibitemOpen
  \bibfield  {author} {\bibinfo {author} {\bibfnamefont {J.}~\bibnamefont
  {Barrett}}, \bibinfo {author} {\bibfnamefont {N.}~\bibnamefont {Linden}},
  \bibinfo {author} {\bibfnamefont {S.}~\bibnamefont {Massar}}, \bibinfo
  {author} {\bibfnamefont {S.}~\bibnamefont {Pironio}}, \bibinfo {author}
  {\bibfnamefont {S.}~\bibnamefont {Popescu}}, \ and\ \bibinfo {author}
  {\bibfnamefont {D.}~\bibnamefont {Roberts}},\ }\bibfield  {title} {\enquote
  {\bibinfo {title} {Nonlocal correlations as an information-theoretic
  resource},}\ }\href {\doibase 10.1103/PhysRevA.71.022101} {\bibfield
  {journal} {\bibinfo  {journal} {Physical Review A}\ }\textbf {\bibinfo
  {volume} {71}},\ \bibinfo {pages} {022101} (\bibinfo {year}
  {2005}{\natexlab{a}})}\BibitemShut {NoStop}%
\bibitem [{\citenamefont {Colbeck}(2007)}]{ColbeckThesis}%
  \BibitemOpen
  \bibfield  {author} {\bibinfo {author} {\bibfnamefont {R.}~\bibnamefont
  {Colbeck}},\ }\emph {\bibinfo {title} {Quantum and Relativistic Protocols For
  Secure Multi-Party Computation}},\ \href@noop {} {Ph.D. thesis},\ \bibinfo
  {school} {University of Cambridge} (\bibinfo {year} {2007}),\ \bibinfo {note}
  {also available as
  \href{https://arxiv.org/abs/0911.3814}{arXiv:0911.3814}.}\BibitemShut {Stop}%
\bibitem [{\citenamefont {Pironio}\ \emph {et~al.}(2010)\citenamefont
  {Pironio}, \citenamefont {Acin}, \citenamefont {Massar}, \citenamefont
  {{Boyer de la Giroday}}, \citenamefont {Matsukevich}, \citenamefont {Maunz},
  \citenamefont {Olmschenk}, \citenamefont {Hayes}, \citenamefont {Luo},
  \citenamefont {Manning},\ and\ \citenamefont {Monroe}}]{PAMBMMOHLMM}%
  \BibitemOpen
  \bibfield  {author} {\bibinfo {author} {\bibfnamefont {S.}~\bibnamefont
  {Pironio}}, \bibinfo {author} {\bibfnamefont {A.}~\bibnamefont {Acin}},
  \bibinfo {author} {\bibfnamefont {S.}~\bibnamefont {Massar}}, \bibinfo
  {author} {\bibfnamefont {A.}~\bibnamefont {{Boyer de la Giroday}}}, \bibinfo
  {author} {\bibfnamefont {D.~N.}\ \bibnamefont {Matsukevich}}, \bibinfo
  {author} {\bibfnamefont {P.}~\bibnamefont {Maunz}}, \bibinfo {author}
  {\bibfnamefont {S.}~\bibnamefont {Olmschenk}}, \bibinfo {author}
  {\bibfnamefont {D.}~\bibnamefont {Hayes}}, \bibinfo {author} {\bibfnamefont
  {L.}~\bibnamefont {Luo}}, \bibinfo {author} {\bibfnamefont {T.~A.}\
  \bibnamefont {Manning}}, \ and\ \bibinfo {author} {\bibfnamefont
  {C.}~\bibnamefont {Monroe}},\ }\bibfield  {title} {\enquote {\bibinfo {title}
  {Random numbers certified by {B}ell's theorem},}\ }\href {\doibase
  10.1038/nature09008} {\bibfield  {journal} {\bibinfo  {journal} {Nature}\
  }\textbf {\bibinfo {volume} {464}},\ \bibinfo {pages} {1021--1024} (\bibinfo
  {year} {2010})}\BibitemShut {NoStop}%
\bibitem [{\citenamefont {Colbeck}\ and\ \citenamefont {Kent}(2011)}]{CK2}%
  \BibitemOpen
  \bibfield  {author} {\bibinfo {author} {\bibfnamefont {R.}~\bibnamefont
  {Colbeck}}\ and\ \bibinfo {author} {\bibfnamefont {A.}~\bibnamefont {Kent}},\
  }\bibfield  {title} {\enquote {\bibinfo {title} {Private randomness expansion
  with untrusted devices},}\ }\href {\doibase 10.1088/1751-8113/44/9/095305}
  {\bibfield  {journal} {\bibinfo  {journal} {Journal of Physics A}\ }\textbf
  {\bibinfo {volume} {44}},\ \bibinfo {pages} {095305} (\bibinfo {year}
  {2011})}\BibitemShut {NoStop}%
\bibitem [{\citenamefont {Miller}\ and\ \citenamefont {Shi}(2014)}]{MS1}%
  \BibitemOpen
  \bibfield  {author} {\bibinfo {author} {\bibfnamefont {C.~A.}\ \bibnamefont
  {Miller}}\ and\ \bibinfo {author} {\bibfnamefont {Y.}~\bibnamefont {Shi}},\
  }\bibfield  {title} {\enquote {\bibinfo {title} {Robust protocols for
  securely expanding randomness and distributing keys using untrusted quantum
  devices},}\ }in\ \href {\doibase 10.1145/2591796.2591843} {\emph {\bibinfo
  {booktitle} {Proceedings of the 46th Annual ACM Symposium on Theory of
  Computing}}},\ \bibinfo {series and number} {STOC '14}\ (\bibinfo
  {publisher} {ACM},\ \bibinfo {address} {New York, NY, USA},\ \bibinfo {year}
  {2014})\ pp.\ \bibinfo {pages} {417--426}\BibitemShut {NoStop}%
\bibitem [{\citenamefont {Miller}\ and\ \citenamefont
  {Shi}(2017{\natexlab{a}})}]{MS2}%
  \BibitemOpen
  \bibfield  {author} {\bibinfo {author} {\bibfnamefont {C.~A.}\ \bibnamefont
  {Miller}}\ and\ \bibinfo {author} {\bibfnamefont {Y.}~\bibnamefont {Shi}},\
  }\bibfield  {title} {\enquote {\bibinfo {title} {Universal security for
  randomness expansion from the spot-checking protocol},}\ }\href {\doibase
  10.1137/15M1044333} {\bibfield  {journal} {\bibinfo  {journal} {SIAM Journal
  of Computing}\ }\textbf {\bibinfo {volume} {46}},\ \bibinfo {pages}
  {1304--1335} (\bibinfo {year} {2017}{\natexlab{a}})}\BibitemShut {NoStop}%
\bibitem [{\citenamefont {Colbeck}\ and\ \citenamefont
  {Renner}(2012)}]{CR_free}%
  \BibitemOpen
  \bibfield  {author} {\bibinfo {author} {\bibfnamefont {R.}~\bibnamefont
  {Colbeck}}\ and\ \bibinfo {author} {\bibfnamefont {R.}~\bibnamefont
  {Renner}},\ }\bibfield  {title} {\enquote {\bibinfo {title} {Free randomness
  can be amplified},}\ }\href {\doibase 10.1038/ncomms2300} {\bibfield
  {journal} {\bibinfo  {journal} {Nature Physics}\ }\textbf {\bibinfo {volume}
  {8}},\ \bibinfo {pages} {450--454} (\bibinfo {year} {2012})}\BibitemShut
  {NoStop}%
\bibitem [{\citenamefont {Ekert}(1991)}]{Ekert}%
  \BibitemOpen
  \bibfield  {author} {\bibinfo {author} {\bibfnamefont {A.~K.}\ \bibnamefont
  {Ekert}},\ }\bibfield  {title} {\enquote {\bibinfo {title} {Quantum
  cryptography based on {B}ell's theorem},}\ }\href {\doibase
  10.1103/PhysRevLett.67.661} {\bibfield  {journal} {\bibinfo  {journal}
  {Physical Review Letters}\ }\textbf {\bibinfo {volume} {67}},\ \bibinfo
  {pages} {661--663} (\bibinfo {year} {1991})}\BibitemShut {NoStop}%
\bibitem [{\citenamefont {Barrett}\ \emph
  {et~al.}(2005{\natexlab{b}})\citenamefont {Barrett}, \citenamefont {Hardy},\
  and\ \citenamefont {Kent}}]{BHK}%
  \BibitemOpen
  \bibfield  {author} {\bibinfo {author} {\bibfnamefont {J.}~\bibnamefont
  {Barrett}}, \bibinfo {author} {\bibfnamefont {L.}~\bibnamefont {Hardy}}, \
  and\ \bibinfo {author} {\bibfnamefont {A.}~\bibnamefont {Kent}},\ }\bibfield
  {title} {\enquote {\bibinfo {title} {No signalling and quantum key
  distribution},}\ }\href {\doibase 10.1103/PhysRevLett.95.010503} {\bibfield
  {journal} {\bibinfo  {journal} {Physical Review Letters}\ }\textbf {\bibinfo
  {volume} {95}},\ \bibinfo {pages} {010503} (\bibinfo {year}
  {2005}{\natexlab{b}})}\BibitemShut {NoStop}%
\bibitem [{\citenamefont {Acin}\ \emph {et~al.}(2007)\citenamefont {Acin},
  \citenamefont {Brunner}, \citenamefont {Gisin}, \citenamefont {Massar},
  \citenamefont {Pironio},\ and\ \citenamefont {Scarani}}]{ABGMPS}%
  \BibitemOpen
  \bibfield  {author} {\bibinfo {author} {\bibfnamefont {A.}~\bibnamefont
  {Acin}}, \bibinfo {author} {\bibfnamefont {N.}~\bibnamefont {Brunner}},
  \bibinfo {author} {\bibfnamefont {N.}~\bibnamefont {Gisin}}, \bibinfo
  {author} {\bibfnamefont {S.}~\bibnamefont {Massar}}, \bibinfo {author}
  {\bibfnamefont {S.}~\bibnamefont {Pironio}}, \ and\ \bibinfo {author}
  {\bibfnamefont {V.}~\bibnamefont {Scarani}},\ }\bibfield  {title} {\enquote
  {\bibinfo {title} {Device-independent security of quantum cryptography
  against collective attacks},}\ }\href {\doibase
  10.1103/PhysRevLett.98.230501} {\bibfield  {journal} {\bibinfo  {journal}
  {Physical Review Letters}\ }\textbf {\bibinfo {volume} {98}},\ \bibinfo
  {pages} {230501} (\bibinfo {year} {2007})}\BibitemShut {NoStop}%
\bibitem [{\citenamefont {Pironio}\ \emph {et~al.}(2009)\citenamefont
  {Pironio}, \citenamefont {Acin}, \citenamefont {Brunner}, \citenamefont
  {Gisin}, \citenamefont {Massar},\ and\ \citenamefont {Scarani}}]{PABGMS}%
  \BibitemOpen
  \bibfield  {author} {\bibinfo {author} {\bibfnamefont {S.}~\bibnamefont
  {Pironio}}, \bibinfo {author} {\bibfnamefont {A.}~\bibnamefont {Acin}},
  \bibinfo {author} {\bibfnamefont {N.}~\bibnamefont {Brunner}}, \bibinfo
  {author} {\bibfnamefont {N.}~\bibnamefont {Gisin}}, \bibinfo {author}
  {\bibfnamefont {S.}~\bibnamefont {Massar}}, \ and\ \bibinfo {author}
  {\bibfnamefont {V.}~\bibnamefont {Scarani}},\ }\bibfield  {title} {\enquote
  {\bibinfo {title} {Device-independent quantum key distribution secure against
  collective attacks},}\ }\href {\doibase 10.1088/1367-2630/11/4/045021}
  {\bibfield  {journal} {\bibinfo  {journal} {New Journal of Physics}\ }\textbf
  {\bibinfo {volume} {11}},\ \bibinfo {pages} {045021} (\bibinfo {year}
  {2009})}\BibitemShut {NoStop}%
\bibitem [{\citenamefont {Vazirani}\ and\ \citenamefont {Vidick}(2014)}]{VV2}%
  \BibitemOpen
  \bibfield  {author} {\bibinfo {author} {\bibfnamefont {U.}~\bibnamefont
  {Vazirani}}\ and\ \bibinfo {author} {\bibfnamefont {T.}~\bibnamefont
  {Vidick}},\ }\bibfield  {title} {\enquote {\bibinfo {title} {Fully
  device-independent quantum key distribution},}\ }\href {\doibase
  10.1103/PhysRevLett.113.140501} {\bibfield  {journal} {\bibinfo  {journal}
  {Physical Review Letters}\ }\textbf {\bibinfo {volume} {113}},\ \bibinfo
  {pages} {140501} (\bibinfo {year} {2014})}\BibitemShut {NoStop}%
\bibitem [{\citenamefont {Arnon-Friedman}\ \emph {et~al.}(2018)\citenamefont
  {Arnon-Friedman}, \citenamefont {Dupuis}, \citenamefont {Fawzi},
  \citenamefont {Renner},\ and\ \citenamefont {Vidick}}]{ADFRV}%
  \BibitemOpen
  \bibfield  {author} {\bibinfo {author} {\bibfnamefont {R.}~\bibnamefont
  {Arnon-Friedman}}, \bibinfo {author} {\bibfnamefont {F.}~\bibnamefont
  {Dupuis}}, \bibinfo {author} {\bibfnamefont {O.}~\bibnamefont {Fawzi}},
  \bibinfo {author} {\bibfnamefont {R.}~\bibnamefont {Renner}}, \ and\ \bibinfo
  {author} {\bibfnamefont {T.}~\bibnamefont {Vidick}},\ }\bibfield  {title}
  {\enquote {\bibinfo {title} {Practical device-independent quantum
  cryptography via entropy accumulation},}\ }\href {\doibase
  10.1038/s41467-017-02307-4} {\bibfield  {journal} {\bibinfo  {journal}
  {Nature communications}\ }\textbf {\bibinfo {volume} {9}},\ \bibinfo {pages}
  {459} (\bibinfo {year} {2018})}\BibitemShut {NoStop}%
\bibitem [{\citenamefont {de~la Torre}\ \emph {et~al.}(2015)\citenamefont
  {de~la Torre}, \citenamefont {Hoban}, \citenamefont {Dhara}, \citenamefont
  {Prettico},\ and\ \citenamefont {Ac\'in}}]{delaTorreMaxNonlocal}%
  \BibitemOpen
  \bibfield  {author} {\bibinfo {author} {\bibfnamefont {G.}~\bibnamefont
  {de~la Torre}}, \bibinfo {author} {\bibfnamefont {M.~J.}\ \bibnamefont
  {Hoban}}, \bibinfo {author} {\bibfnamefont {C.}~\bibnamefont {Dhara}},
  \bibinfo {author} {\bibfnamefont {G.}~\bibnamefont {Prettico}}, \ and\
  \bibinfo {author} {\bibfnamefont {A.}~\bibnamefont {Ac\'in}},\ }\bibfield
  {title} {\enquote {\bibinfo {title} {Maximally nonlocal theories cannot be
  maximally random},}\ }\href {\doibase 10.1103/physrevlett.114.160502}
  {\bibfield  {journal} {\bibinfo  {journal} {Physical Review Letters}\
  }\textbf {\bibinfo {volume} {114}},\ \bibinfo {pages} {160502} (\bibinfo
  {year} {2015})}\BibitemShut {NoStop}%
\bibitem [{\citenamefont {Ac\'in}\ \emph {et~al.}(2012)\citenamefont {Ac\'in},
  \citenamefont {Massar},\ and\ \citenamefont
  {Pironio}}]{AcinRandomnessNonlocality}%
  \BibitemOpen
  \bibfield  {author} {\bibinfo {author} {\bibfnamefont {A.}~\bibnamefont
  {Ac\'in}}, \bibinfo {author} {\bibfnamefont {S.}~\bibnamefont {Massar}}, \
  and\ \bibinfo {author} {\bibfnamefont {S.}~\bibnamefont {Pironio}},\
  }\bibfield  {title} {\enquote {\bibinfo {title} {Randomness versus
  nonlocality and entanglement},}\ }\href {\doibase
  10.1103/PhysRevLett.108.100402} {\bibfield  {journal} {\bibinfo  {journal}
  {Physical Review Letters}\ }\textbf {\bibinfo {volume} {108}},\ \bibinfo
  {pages} {100402} (\bibinfo {year} {2012})}\BibitemShut {NoStop}%
\bibitem [{\citenamefont {Dhara}\ \emph {et~al.}(2013)\citenamefont {Dhara},
  \citenamefont {Prettico},\ and\ \citenamefont {Ac\'in}}]{DharaMaxRand}%
  \BibitemOpen
  \bibfield  {author} {\bibinfo {author} {\bibfnamefont {C.}~\bibnamefont
  {Dhara}}, \bibinfo {author} {\bibfnamefont {G.}~\bibnamefont {Prettico}}, \
  and\ \bibinfo {author} {\bibfnamefont {A.}~\bibnamefont {Ac\'in}},\
  }\bibfield  {title} {\enquote {\bibinfo {title} {Maximal quantum randomness
  in {B}ell tests},}\ }\href {\doibase 10.1103/PhysRevA.88.052116} {\bibfield
  {journal} {\bibinfo  {journal} {Physical Review A}\ }\textbf {\bibinfo
  {volume} {88}},\ \bibinfo {pages} {052116} (\bibinfo {year}
  {2013})}\BibitemShut {NoStop}%
\bibitem [{\citenamefont {Ac\'in}\ \emph {et~al.}(2016)\citenamefont {Ac\'in},
  \citenamefont {Pironio}, \citenamefont {V\'ertesi},\ and\ \citenamefont
  {Wittek}}]{AcinOptimalEbit}%
  \BibitemOpen
  \bibfield  {author} {\bibinfo {author} {\bibfnamefont {A.}~\bibnamefont
  {Ac\'in}}, \bibinfo {author} {\bibfnamefont {S.}~\bibnamefont {Pironio}},
  \bibinfo {author} {\bibfnamefont {T.}~\bibnamefont {V\'ertesi}}, \ and\
  \bibinfo {author} {\bibfnamefont {P.}~\bibnamefont {Wittek}},\ }\bibfield
  {title} {\enquote {\bibinfo {title} {Optimal randomness certification from
  one entangled bit},}\ }\href {\doibase 10.1103/PhysRevA.93.040102} {\bibfield
   {journal} {\bibinfo  {journal} {Physical Review A}\ }\textbf {\bibinfo
  {volume} {93}},\ \bibinfo {pages} {040102} (\bibinfo {year}
  {2016})}\BibitemShut {NoStop}%
\bibitem [{\citenamefont {Law}\ \emph {et~al.}(2014)\citenamefont {Law},
  \citenamefont {Thinh}, \citenamefont {Bancal},\ and\ \citenamefont
  {Scarani}}]{Law_2014}%
  \BibitemOpen
  \bibfield  {author} {\bibinfo {author} {\bibfnamefont {Y.~Z.}\ \bibnamefont
  {Law}}, \bibinfo {author} {\bibfnamefont {L.~P.}\ \bibnamefont {Thinh}},
  \bibinfo {author} {\bibfnamefont {J.-D.}\ \bibnamefont {Bancal}}, \ and\
  \bibinfo {author} {\bibfnamefont {V.}~\bibnamefont {Scarani}},\ }\bibfield
  {title} {\enquote {\bibinfo {title} {Quantum randomness extraction for
  various levels of characterization of the devices},}\ }\href {\doibase
  10.1088/1751-8113/47/42/424028} {\bibfield  {journal} {\bibinfo  {journal}
  {Journal of Physics A: Mathematical and Theoretical}\ }\textbf {\bibinfo
  {volume} {47}},\ \bibinfo {pages} {424028} (\bibinfo {year}
  {2014})}\BibitemShut {NoStop}%
\bibitem [{\citenamefont {Andersson}\ \emph {et~al.}(2018)\citenamefont
  {Andersson}, \citenamefont {Badzi{\k{a}}g}, \citenamefont {Dumitru},\ and\
  \citenamefont {Cabello}}]{Andersson18}%
  \BibitemOpen
  \bibfield  {author} {\bibinfo {author} {\bibfnamefont {O.}~\bibnamefont
  {Andersson}}, \bibinfo {author} {\bibfnamefont {P.}~\bibnamefont
  {Badzi{\k{a}}g}}, \bibinfo {author} {\bibfnamefont {I.}~\bibnamefont
  {Dumitru}}, \ and\ \bibinfo {author} {\bibfnamefont {A.}~\bibnamefont
  {Cabello}},\ }\bibfield  {title} {\enquote {\bibinfo {title}
  {Device-independent certification of two bits of randomness from one
  entangled bit and gisin's elegant bell inequality},}\ }\href {\doibase
  10.1103/PhysRevA.97.012314} {\bibfield  {journal} {\bibinfo  {journal} {Phys.
  Rev. A}\ }\textbf {\bibinfo {volume} {97}},\ \bibinfo {pages} {012314}
  (\bibinfo {year} {2018})}\BibitemShut {NoStop}%
\bibitem [{\citenamefont {Brown}\ \emph {et~al.}(2020)\citenamefont {Brown},
  \citenamefont {Ragy},\ and\ \citenamefont {Colbeck}}]{BRC}%
  \BibitemOpen
  \bibfield  {author} {\bibinfo {author} {\bibfnamefont {P.~J.}\ \bibnamefont
  {Brown}}, \bibinfo {author} {\bibfnamefont {S.}~\bibnamefont {Ragy}}, \ and\
  \bibinfo {author} {\bibfnamefont {R.}~\bibnamefont {Colbeck}},\ }\bibfield
  {title} {\enquote {\bibinfo {title} {A framework for quantum-secure
  device-independent randomness expansion},}\ }\href {\doibase
  10.1109/tit.2019.2960252} {\bibfield  {journal} {\bibinfo  {journal} {IEEE
  Transactions on Information Theory}\ }\textbf {\bibinfo {volume} {66}},\
  \bibinfo {pages} {2964–2987} (\bibinfo {year} {2020})}\BibitemShut
  {NoStop}%
\bibitem [{\citenamefont {Woodhead}\ \emph {et~al.}(2020)\citenamefont
  {Woodhead}, \citenamefont {Kaniewski}, \citenamefont {Bourdoncle},
  \citenamefont {Salavrakos}, \citenamefont {Bowles}, \citenamefont {Ac\'in},\
  and\ \citenamefont {Augusiak}}]{WoodheadMaxRandomness}%
  \BibitemOpen
  \bibfield  {author} {\bibinfo {author} {\bibfnamefont {E.}~\bibnamefont
  {Woodhead}}, \bibinfo {author} {\bibfnamefont {J.}~\bibnamefont {Kaniewski}},
  \bibinfo {author} {\bibfnamefont {B.}~\bibnamefont {Bourdoncle}}, \bibinfo
  {author} {\bibfnamefont {A.}~\bibnamefont {Salavrakos}}, \bibinfo {author}
  {\bibfnamefont {J.}~\bibnamefont {Bowles}}, \bibinfo {author} {\bibfnamefont
  {A.}~\bibnamefont {Ac\'in}}, \ and\ \bibinfo {author} {\bibfnamefont
  {R.}~\bibnamefont {Augusiak}},\ }\bibfield  {title} {\enquote {\bibinfo
  {title} {Maximal randomness from partially entangled states},}\ }\href
  {\doibase 10.1103/PhysRevResearch.2.042028} {\bibfield  {journal} {\bibinfo
  {journal} {Physical Review Research}\ }\textbf {\bibinfo {volume} {2}},\
  \bibinfo {pages} {042028} (\bibinfo {year} {2020})}\BibitemShut {NoStop}%
\bibitem [{\citenamefont {Bamps}\ and\ \citenamefont
  {Pironio}(2015)}]{BampsPironio}%
  \BibitemOpen
  \bibfield  {author} {\bibinfo {author} {\bibfnamefont {C.}~\bibnamefont
  {Bamps}}\ and\ \bibinfo {author} {\bibfnamefont {S.}~\bibnamefont
  {Pironio}},\ }\bibfield  {title} {\enquote {\bibinfo {title} {Sum-of-squares
  decompositions for a family of {C}lauser-{H}orne-{S}himony-{H}olt-like
  inequalities and their application to self-testing},}\ }\href {\doibase
  10.1103/PhysRevA.91.052111} {\bibfield  {journal} {\bibinfo  {journal}
  {Physical Review A}\ }\textbf {\bibinfo {volume} {91}},\ \bibinfo {pages}
  {052111} (\bibinfo {year} {2015})}\BibitemShut {NoStop}%
\bibitem [{\citenamefont {Bhavsar}\ \emph {et~al.}(2021)\citenamefont
  {Bhavsar}, \citenamefont {Ragy},\ and\ \citenamefont {Colbeck}}]{BhavsarDI}%
  \BibitemOpen
  \bibfield  {author} {\bibinfo {author} {\bibfnamefont {R.}~\bibnamefont
  {Bhavsar}}, \bibinfo {author} {\bibfnamefont {S.}~\bibnamefont {Ragy}}, \
  and\ \bibinfo {author} {\bibfnamefont {R.}~\bibnamefont {Colbeck}},\
  }\href@noop {} {\enquote {\bibinfo {title} {Improved device-independent
  randomness expansion rates from tight bounds on the two sided randomness
  using {CHSH} tests},}\ } (\bibinfo {year} {2021}),\ \Eprint
  {http://arxiv.org/abs/2103.07504} {arXiv:2103.07504 [quant-ph]} \BibitemShut
  {NoStop}%
\bibitem [{\citenamefont {Brown}\ \emph {et~al.}(2021)\citenamefont {Brown},
  \citenamefont {Fawzi},\ and\ \citenamefont
  {Fawzi}}]{BrownDeviceIndependent2}%
  \BibitemOpen
  \bibfield  {author} {\bibinfo {author} {\bibfnamefont {P.}~\bibnamefont
  {Brown}}, \bibinfo {author} {\bibfnamefont {H.}~\bibnamefont {Fawzi}}, \ and\
  \bibinfo {author} {\bibfnamefont {O.}~\bibnamefont {Fawzi}},\ }\href@noop {}
  {\enquote {\bibinfo {title} {Device-independent lower bounds on the
  conditional von {N}eumann entropy},}\ } (\bibinfo {year} {2021}),\ \Eprint
  {http://arxiv.org/abs/2106.13692} {arXiv:2106.13692 [quant-ph]} \BibitemShut
  {NoStop}%
\bibitem [{\citenamefont {Werner}(1989)}]{Werner}%
  \BibitemOpen
  \bibfield  {author} {\bibinfo {author} {\bibfnamefont {R.~F.}\ \bibnamefont
  {Werner}},\ }\bibfield  {title} {\enquote {\bibinfo {title} {Quantum states
  with {E}instein-{P}odolsky-{R}osen correlations admitting a hidden-variable
  model},}\ }\href {\doibase 10.1103/PhysRevA.40.4277} {\bibfield  {journal}
  {\bibinfo  {journal} {Physical Review A}\ }\textbf {\bibinfo {volume} {40}},\
  \bibinfo {pages} {4277--4281} (\bibinfo {year} {1989})}\BibitemShut {NoStop}%
\bibitem [{\citenamefont {Paulsen}(2003)}]{paulsen_2003}%
  \BibitemOpen
  \bibfield  {author} {\bibinfo {author} {\bibfnamefont {V.}~\bibnamefont
  {Paulsen}},\ }\href {\doibase 10.1017/CBO9780511546631} {\emph {\bibinfo
  {title} {Completely Bounded Maps and Operator Algebras}}},\ Cambridge Studies
  in Advanced Mathematics\ (\bibinfo  {publisher} {Cambridge University
  Press},\ \bibinfo {year} {2003})\BibitemShut {NoStop}%
\bibitem [{\citenamefont {Cirel'son}(1980)}]{Cirelson}%
  \BibitemOpen
  \bibfield  {author} {\bibinfo {author} {\bibfnamefont {B.}~\bibnamefont
  {Cirel'son}},\ }\bibfield  {title} {\enquote {\bibinfo {title} {Quantum
  generalizations of {B}ell's inequality},}\ }\href {\doibase
  10.1007/BF00417500} {\bibfield  {journal} {\bibinfo  {journal} {Letters in
  Mathematical Physics}\ }\textbf {\bibinfo {volume} {4}},\ \bibinfo {pages}
  {93--100} (\bibinfo {year} {1980})}\BibitemShut {NoStop}%
\bibitem [{\citenamefont {Popescu}\ and\ \citenamefont
  {Rohrlich}(1992)}]{PopescuRohrlich}%
  \BibitemOpen
  \bibfield  {author} {\bibinfo {author} {\bibfnamefont {S.}~\bibnamefont
  {Popescu}}\ and\ \bibinfo {author} {\bibfnamefont {D.}~\bibnamefont
  {Rohrlich}},\ }\bibfield  {title} {\enquote {\bibinfo {title} {Which states
  violate {B}ell's inequality maximally?}}\ }\href {\doibase
  10.1016/0375-9601(92)90819-8} {\bibfield  {journal} {\bibinfo  {journal}
  {Physics Letters A}\ }\textbf {\bibinfo {volume} {169}},\ \bibinfo {pages}
  {411--414} (\bibinfo {year} {1992})}\BibitemShut {NoStop}%
\bibitem [{\citenamefont {Bardyn}\ \emph {et~al.}(2009)\citenamefont {Bardyn},
  \citenamefont {Liew}, \citenamefont {Massar}, \citenamefont {McKague},\ and\
  \citenamefont {Scarani}}]{BardynSelfTest}%
  \BibitemOpen
  \bibfield  {author} {\bibinfo {author} {\bibfnamefont {C.-E.}\ \bibnamefont
  {Bardyn}}, \bibinfo {author} {\bibfnamefont {T.~C.~H.}\ \bibnamefont {Liew}},
  \bibinfo {author} {\bibfnamefont {S.}~\bibnamefont {Massar}}, \bibinfo
  {author} {\bibfnamefont {M.}~\bibnamefont {McKague}}, \ and\ \bibinfo
  {author} {\bibfnamefont {V.}~\bibnamefont {Scarani}},\ }\bibfield  {title}
  {\enquote {\bibinfo {title} {Device-independent state estimation based on
  {B}ell's inequalities},}\ }\href {\doibase 10.1103/PhysRevA.80.062327}
  {\bibfield  {journal} {\bibinfo  {journal} {Physical Review A}\ }\textbf
  {\bibinfo {volume} {80}},\ \bibinfo {pages} {062327} (\bibinfo {year}
  {2009})}\BibitemShut {NoStop}%
\bibitem [{\citenamefont {Dupuis}\ \emph {et~al.}(2016)\citenamefont {Dupuis},
  \citenamefont {Fawzi},\ and\ \citenamefont {Renner}}]{DFR}%
  \BibitemOpen
  \bibfield  {author} {\bibinfo {author} {\bibfnamefont {F.}~\bibnamefont
  {Dupuis}}, \bibinfo {author} {\bibfnamefont {O.}~\bibnamefont {Fawzi}}, \
  and\ \bibinfo {author} {\bibfnamefont {R.}~\bibnamefont {Renner}},\
  }\href@noop {} {\enquote {\bibinfo {title} {Entropy accumulation},}\
  }\bibinfo {howpublished} {e-print
  \href{https://arxiv.org/abs/1607.01796}{arXiv:1607.01796}} (\bibinfo {year}
  {2016})\BibitemShut {NoStop}%
\bibitem [{\citenamefont {Dupuis}\ and\ \citenamefont {Fawzi}(2019)}]{EAT2}%
  \BibitemOpen
  \bibfield  {author} {\bibinfo {author} {\bibfnamefont {F.}~\bibnamefont
  {Dupuis}}\ and\ \bibinfo {author} {\bibfnamefont {O.}~\bibnamefont {Fawzi}},\
  }\bibfield  {title} {\enquote {\bibinfo {title} {Entropy accumulation with
  improved second-order term},}\ }\href {\doibase 10.1109/tit.2019.2929564}
  {\bibfield  {journal} {\bibinfo  {journal} {{IEEE} Transactions on
  Information Theory}\ }\textbf {\bibinfo {volume} {65}},\ \bibinfo {pages}
  {7596--7612} (\bibinfo {year} {2019})}\BibitemShut {NoStop}%
\bibitem [{\citenamefont {Liu}\ \emph {et~al.}(2021)\citenamefont {Liu},
  \citenamefont {Li}, \citenamefont {Ragy}, \citenamefont {Zhao}, \citenamefont
  {Bai}, \citenamefont {Liu}, \citenamefont {Brown}, \citenamefont {Zhang},
  \citenamefont {Colbeck}, \citenamefont {Fan}, \citenamefont {Zhang},\ and\
  \citenamefont {Pan}}]{LLR&}%
  \BibitemOpen
  \bibfield  {author} {\bibinfo {author} {\bibfnamefont {W.-Z.}\ \bibnamefont
  {Liu}}, \bibinfo {author} {\bibfnamefont {M.-H.}\ \bibnamefont {Li}},
  \bibinfo {author} {\bibfnamefont {S.}~\bibnamefont {Ragy}}, \bibinfo {author}
  {\bibfnamefont {S.-R.}\ \bibnamefont {Zhao}}, \bibinfo {author}
  {\bibfnamefont {B.}~\bibnamefont {Bai}}, \bibinfo {author} {\bibfnamefont
  {Y.}~\bibnamefont {Liu}}, \bibinfo {author} {\bibfnamefont {P.~J.}\
  \bibnamefont {Brown}}, \bibinfo {author} {\bibfnamefont {J.}~\bibnamefont
  {Zhang}}, \bibinfo {author} {\bibfnamefont {R.}~\bibnamefont {Colbeck}},
  \bibinfo {author} {\bibfnamefont {J.}~\bibnamefont {Fan}}, \bibinfo {author}
  {\bibfnamefont {Q.}~\bibnamefont {Zhang}}, \ and\ \bibinfo {author}
  {\bibfnamefont {J.-W.}\ \bibnamefont {Pan}},\ }\bibfield  {title} {\enquote
  {\bibinfo {title} {Device-independent randomness expansion against quantum
  side information},}\ }\href {\doibase 10.1038/s41567-020-01147-2} {\bibfield
  {journal} {\bibinfo  {journal} {Nature Physics}\ }\textbf {\bibinfo {volume}
  {17}},\ \bibinfo {pages} {448--451} (\bibinfo {year} {2021})}\BibitemShut
  {NoStop}%
\bibitem [{\citenamefont {Wooltorton}\ \emph {et~al.}()\citenamefont
  {Wooltorton}, \citenamefont {Brown},\ and\ \citenamefont {Colbeck}}]{supp}%
  \BibitemOpen
  \bibfield  {author} {\bibinfo {author} {\bibfnamefont {L.}~\bibnamefont
  {Wooltorton}}, \bibinfo {author} {\bibfnamefont {P.}~\bibnamefont {Brown}}, \
  and\ \bibinfo {author} {\bibfnamefont {R.}~\bibnamefont {Colbeck}},\
  }\href@noop {} {}\bibinfo {note} {Supplemental Material containing proofs of
  the propositions and including additional
  references~\cite{Jordan,GohGeometry,SekatskiBuilidngBlocks,ValcarceSelfTest}}\BibitemShut
  {NoStop}%
\bibitem [{\citenamefont {Miller}\ and\ \citenamefont
  {Shi}(2017{\natexlab{b}})}]{MillerBlind}%
  \BibitemOpen
  \bibfield  {author} {\bibinfo {author} {\bibfnamefont {C.~A.}\ \bibnamefont
  {Miller}}\ and\ \bibinfo {author} {\bibfnamefont {Y.}~\bibnamefont {Shi}},\
  }\bibfield  {title} {\enquote {\bibinfo {title} {Randomness in nonlocal games
  between mistrustful players},}\ }\href
  {https://pubmed.ncbi.nlm.nih.gov/29643748} {\bibfield  {journal} {\bibinfo
  {journal} {Quantum information {\&} computation}\ }\textbf {\bibinfo {volume}
  {17}},\ \bibinfo {pages} {595--610} (\bibinfo {year}
  {2017}{\natexlab{b}})}\BibitemShut {NoStop}%
\bibitem [{\citenamefont {Fu}\ and\ \citenamefont
  {Miller}(2018)}]{HonghaoBlind}%
  \BibitemOpen
  \bibfield  {author} {\bibinfo {author} {\bibfnamefont {H.}~\bibnamefont
  {Fu}}\ and\ \bibinfo {author} {\bibfnamefont {C.~A.}\ \bibnamefont
  {Miller}},\ }\bibfield  {title} {\enquote {\bibinfo {title} {Local
  randomness: Examples and application},}\ }\href {\doibase
  10.1103/PhysRevA.97.032324} {\bibfield  {journal} {\bibinfo  {journal}
  {Physical Review A}\ }\textbf {\bibinfo {volume} {97}},\ \bibinfo {pages}
  {032324} (\bibinfo {year} {2018})}\BibitemShut {NoStop}%
\bibitem [{\citenamefont {Metger}\ \emph {et~al.}(2022)\citenamefont {Metger},
  \citenamefont {Fawzi}, \citenamefont {Sutter},\ and\ \citenamefont
  {Renner}}]{MetgerGEAT}%
  \BibitemOpen
  \bibfield  {author} {\bibinfo {author} {\bibfnamefont {T.}~\bibnamefont
  {Metger}}, \bibinfo {author} {\bibfnamefont {O.}~\bibnamefont {Fawzi}},
  \bibinfo {author} {\bibfnamefont {D.}~\bibnamefont {Sutter}}, \ and\ \bibinfo
  {author} {\bibfnamefont {R.}~\bibnamefont {Renner}},\ }\bibfield  {title}
  {\enquote {\bibinfo {title} {Generalised entropy accumulation},}\ }\href@noop
  {} {\  (\bibinfo {year} {2022})},\ \Eprint {http://arxiv.org/abs/2203.04989}
  {arXiv:2203.04989 [quant-ph]} \BibitemShut {NoStop}%
\bibitem [{\citenamefont {Woodhead}\ \emph {et~al.}(2018)\citenamefont
  {Woodhead}, \citenamefont {Bourdoncle},\ and\ \citenamefont
  {Acín}}]{WoodheadMermin}%
  \BibitemOpen
  \bibfield  {author} {\bibinfo {author} {\bibfnamefont {E.}~\bibnamefont
  {Woodhead}}, \bibinfo {author} {\bibfnamefont {B.}~\bibnamefont
  {Bourdoncle}}, \ and\ \bibinfo {author} {\bibfnamefont {A.}~\bibnamefont
  {Acín}},\ }\bibfield  {title} {\enquote {\bibinfo {title} {Randomness versus
  nonlocality in the {M}ermin-{B}ell experiment with three parties},}\ }\href
  {\doibase 10.22331/q-2018-08-17-82} {\bibfield  {journal} {\bibinfo
  {journal} {Quantum}\ }\textbf {\bibinfo {volume} {2}},\ \bibinfo {pages} {82}
  (\bibinfo {year} {2018})}\BibitemShut {NoStop}%
\bibitem [{\citenamefont {Grasselli}\ \emph {et~al.}(2021)\citenamefont
  {Grasselli}, \citenamefont {Murta}, \citenamefont {Kampermann},\ and\
  \citenamefont {Bru\ss{}}}]{GrasselliMulti}%
  \BibitemOpen
  \bibfield  {author} {\bibinfo {author} {\bibfnamefont {F.}~\bibnamefont
  {Grasselli}}, \bibinfo {author} {\bibfnamefont {G.}~\bibnamefont {Murta}},
  \bibinfo {author} {\bibfnamefont {H.}~\bibnamefont {Kampermann}}, \ and\
  \bibinfo {author} {\bibfnamefont {D.}~\bibnamefont {Bru\ss{}}},\ }\bibfield
  {title} {\enquote {\bibinfo {title} {Entropy bounds for multiparty
  device-independent cryptography},}\ }\href {\doibase
  10.1103/PRXQuantum.2.010308} {\bibfield  {journal} {\bibinfo  {journal} {PRX
  Quantum}\ }\textbf {\bibinfo {volume} {2}},\ \bibinfo {pages} {010308}
  (\bibinfo {year} {2021})}\BibitemShut {NoStop}%
\bibitem [{\citenamefont {Sarkar}\ \emph {et~al.}(2021)\citenamefont {Sarkar},
  \citenamefont {Saha}, \citenamefont {Kaniewski},\ and\ \citenamefont
  {Augusiak}}]{SarkarSelfTest}%
  \BibitemOpen
  \bibfield  {author} {\bibinfo {author} {\bibfnamefont {S.}~\bibnamefont
  {Sarkar}}, \bibinfo {author} {\bibfnamefont {D.}~\bibnamefont {Saha}},
  \bibinfo {author} {\bibfnamefont {J.}~\bibnamefont {Kaniewski}}, \ and\
  \bibinfo {author} {\bibfnamefont {R.}~\bibnamefont {Augusiak}},\ }\bibfield
  {title} {\enquote {\bibinfo {title} {Self-testing quantum systems of
  arbitrary local dimension with minimal number of measurements},}\ }\href
  {\doibase 10.1038/s41534-021-00490-3} {\bibfield  {journal} {\bibinfo
  {journal} {npj Quantum Information}\ }\textbf {\bibinfo {volume} {7}},\
  \bibinfo {pages} {151} (\bibinfo {year} {2021})}\BibitemShut {NoStop}%
\bibitem [{\citenamefont {Collins}\ \emph {et~al.}(2002)\citenamefont
  {Collins}, \citenamefont {Gisin}, \citenamefont {Popescu}, \citenamefont
  {Roberts},\ and\ \citenamefont {Scarani}}]{CollinesMulti}%
  \BibitemOpen
  \bibfield  {author} {\bibinfo {author} {\bibfnamefont {D.}~\bibnamefont
  {Collins}}, \bibinfo {author} {\bibfnamefont {N.}~\bibnamefont {Gisin}},
  \bibinfo {author} {\bibfnamefont {S.}~\bibnamefont {Popescu}}, \bibinfo
  {author} {\bibfnamefont {D.}~\bibnamefont {Roberts}}, \ and\ \bibinfo
  {author} {\bibfnamefont {V.}~\bibnamefont {Scarani}},\ }\bibfield  {title}
  {\enquote {\bibinfo {title} {Bell-type inequalities to detect true
  $\mathit{n}$-body nonseparability},}\ }\href {\doibase
  10.1103/PhysRevLett.88.170405} {\bibfield  {journal} {\bibinfo  {journal}
  {Physical Review Letters}\ }\textbf {\bibinfo {volume} {88}},\ \bibinfo
  {pages} {170405} (\bibinfo {year} {2002})}\BibitemShut {NoStop}%
\bibitem [{\citenamefont {Masini}\ \emph {et~al.}(2021)\citenamefont {Masini},
  \citenamefont {Pironio},\ and\ \citenamefont {Woodhead}}]{MasiniDI}%
  \BibitemOpen
  \bibfield  {author} {\bibinfo {author} {\bibfnamefont {M.}~\bibnamefont
  {Masini}}, \bibinfo {author} {\bibfnamefont {S.}~\bibnamefont {Pironio}}, \
  and\ \bibinfo {author} {\bibfnamefont {E.}~\bibnamefont {Woodhead}},\
  }\href@noop {} {\enquote {\bibinfo {title} {Simple and practical {DIQKD}
  security analysis via {BB}84-type uncertainty relations and pauli correlation
  constraints},}\ } (\bibinfo {year} {2021}),\ \Eprint
  {http://arxiv.org/abs/2107.08894} {arXiv:2107.08894 [quant-ph]} \BibitemShut
  {NoStop}%
\bibitem [{\citenamefont {Mckague}(2011)}]{McKagueGraph}%
  \BibitemOpen
  \bibfield  {author} {\bibinfo {author} {\bibfnamefont {M.}~\bibnamefont
  {Mckague}},\ }\bibfield  {title} {\enquote {\bibinfo {title} {Self-testing
  graph states},}\ }in\ \href {\doibase 10.1007/978-3-642-54429-3_7} {\emph
  {\bibinfo {booktitle} {Revised Selected Papers of the 6th Conference on
  Theory of Quantum Computation, Communication, and Cryptography - Volume
  6745}}},\ \bibinfo {series and number} {TQC 2011}\ (\bibinfo  {publisher}
  {Springer-Verlag},\ \bibinfo {address} {Berlin, Heidelberg},\ \bibinfo {year}
  {2011})\ p.\ \bibinfo {pages} {104–120}\BibitemShut {NoStop}%
\bibitem [{\citenamefont {Wu}\ \emph {et~al.}(2014)\citenamefont {Wu},
  \citenamefont {Cai}, \citenamefont {Yang}, \citenamefont {Le}, \citenamefont
  {Bancal},\ and\ \citenamefont {Scarani}}]{WuWselfTest}%
  \BibitemOpen
  \bibfield  {author} {\bibinfo {author} {\bibfnamefont {X.}~\bibnamefont
  {Wu}}, \bibinfo {author} {\bibfnamefont {Y.}~\bibnamefont {Cai}}, \bibinfo
  {author} {\bibfnamefont {T.~H.}\ \bibnamefont {Yang}}, \bibinfo {author}
  {\bibfnamefont {H.~N.}\ \bibnamefont {Le}}, \bibinfo {author} {\bibfnamefont
  {J.-D.}\ \bibnamefont {Bancal}}, \ and\ \bibinfo {author} {\bibfnamefont
  {V.}~\bibnamefont {Scarani}},\ }\bibfield  {title} {\enquote {\bibinfo
  {title} {Robust self-testing of the three-qubit \textit{W} state},}\ }\href
  {\doibase 10.1103/PhysRevA.90.042339} {\bibfield  {journal} {\bibinfo
  {journal} {Physical Review A}\ }\textbf {\bibinfo {volume} {90}},\ \bibinfo
  {pages} {042339} (\bibinfo {year} {2014})}\BibitemShut {NoStop}%
\bibitem [{\citenamefont {Jordan}(1875)}]{Jordan}%
  \BibitemOpen
  \bibfield  {author} {\bibinfo {author} {\bibfnamefont {C.}~\bibnamefont
  {Jordan}},\ }\bibfield  {title} {\enquote {\bibinfo {title} {Essai sur la
  g\'eom\'etrie \`a n dimensions},}\ }\href {\doibase 10.24033/bsmf.90}
  {\bibfield  {journal} {\bibinfo  {journal} {Bulletin de la S. M. F.}\
  }\textbf {\bibinfo {volume} {3}},\ \bibinfo {pages} {103--174} (\bibinfo
  {year} {1875})}\BibitemShut {NoStop}%
\bibitem [{\citenamefont {Goh}\ \emph {et~al.}(2018)\citenamefont {Goh},
  \citenamefont {Kaniewski}, \citenamefont {Wolfe}, \citenamefont {V\'ertesi},
  \citenamefont {Wu}, \citenamefont {Cai}, \citenamefont {Liang},\ and\
  \citenamefont {Scarani}}]{GohGeometry}%
  \BibitemOpen
  \bibfield  {author} {\bibinfo {author} {\bibfnamefont {K.~T.}\ \bibnamefont
  {Goh}}, \bibinfo {author} {\bibfnamefont {J.}~\bibnamefont {Kaniewski}},
  \bibinfo {author} {\bibfnamefont {E.}~\bibnamefont {Wolfe}}, \bibinfo
  {author} {\bibfnamefont {T.}~\bibnamefont {V\'ertesi}}, \bibinfo {author}
  {\bibfnamefont {X.}~\bibnamefont {Wu}}, \bibinfo {author} {\bibfnamefont
  {Y.}~\bibnamefont {Cai}}, \bibinfo {author} {\bibfnamefont {Y.-C.}\
  \bibnamefont {Liang}}, \ and\ \bibinfo {author} {\bibfnamefont
  {V.}~\bibnamefont {Scarani}},\ }\bibfield  {title} {\enquote {\bibinfo
  {title} {Geometry of the set of quantum correlations},}\ }\href {\doibase
  10.1103/PhysRevA.97.022104} {\bibfield  {journal} {\bibinfo  {journal}
  {Physical Review A}\ }\textbf {\bibinfo {volume} {97}},\ \bibinfo {pages}
  {022104} (\bibinfo {year} {2018})}\BibitemShut {NoStop}%
\bibitem [{\citenamefont {Sekatski}\ \emph {et~al.}(2018)\citenamefont
  {Sekatski}, \citenamefont {Bancal}, \citenamefont {Wagner},\ and\
  \citenamefont {Sangouard}}]{SekatskiBuilidngBlocks}%
  \BibitemOpen
  \bibfield  {author} {\bibinfo {author} {\bibfnamefont {P.}~\bibnamefont
  {Sekatski}}, \bibinfo {author} {\bibfnamefont {J.-D.}\ \bibnamefont
  {Bancal}}, \bibinfo {author} {\bibfnamefont {S.}~\bibnamefont {Wagner}}, \
  and\ \bibinfo {author} {\bibfnamefont {N.}~\bibnamefont {Sangouard}},\
  }\bibfield  {title} {\enquote {\bibinfo {title} {Certifying the building
  blocks of quantum computers from {B}ell's theorem},}\ }\href {\doibase
  10.1103/PhysRevLett.121.180505} {\bibfield  {journal} {\bibinfo  {journal}
  {Physical Review Letters}\ }\textbf {\bibinfo {volume} {121}},\ \bibinfo
  {pages} {180505} (\bibinfo {year} {2018})}\BibitemShut {NoStop}%
\bibitem [{\citenamefont {Valcarce}\ \emph {et~al.}(2022)\citenamefont
  {Valcarce}, \citenamefont {Zivy}, \citenamefont {Sangouard},\ and\
  \citenamefont {Sekatski}}]{ValcarceSelfTest}%
  \BibitemOpen
  \bibfield  {author} {\bibinfo {author} {\bibfnamefont {X.}~\bibnamefont
  {Valcarce}}, \bibinfo {author} {\bibfnamefont {J.}~\bibnamefont {Zivy}},
  \bibinfo {author} {\bibfnamefont {N.}~\bibnamefont {Sangouard}}, \ and\
  \bibinfo {author} {\bibfnamefont {P.}~\bibnamefont {Sekatski}},\ }\bibfield
  {title} {\enquote {\bibinfo {title} {Self-testing two-qubit maximally
  entangled states from generalized {C}lauser-{H}orne-{S}himony-{H}olt
  tests},}\ }\href {\doibase 10.1103/PhysRevResearch.4.013049} {\bibfield
  {journal} {\bibinfo  {journal} {Physical Review Research}\ }\textbf {\bibinfo
  {volume} {4}},\ \bibinfo {pages} {013049} (\bibinfo {year}
  {2022})}\BibitemShut {NoStop}%
\end{thebibliography}
\end{document}